\documentclass[superscriptaddress,letterpaper]{revtex4}

\usepackage{dsfont,amsthm,amssymb,amsmath}
\usepackage{graphicx,color,tikz}
\usepackage{hyperref}
\usepackage{mathtools}
\usepackage{float}
\usepackage{upgreek, bbold, float, caption}
\usepackage[boxed]{algorithm2e}
\usetikzlibrary{plotmarks}

\begin{document}

\title{Nonlocal games, synchronous correlations, and Bell inequalities}
\author{Nishant Rodrigues}
\affiliation{Department of Computer Science, University of Maryland, College Park, MD}
\affiliation{Joint Center for Quantum Information and Computer Science, College Park, MD}
\author{Brad Lackey}
\affiliation{Department of Computer Science, University of Maryland, College Park, MD}
\affiliation{Department of Mathematics, University of Maryland, College Park, MD}
\affiliation{Institute for Advanced Computer Studies, University of Maryland, College Park, MD}
\affiliation{Quantum Systems Group, Microsoft Quantum, Redmond, WA}

\date{\today}

\newtheorem{theorem}{Theorem}
\newtheorem{lemma}[theorem]{Lemma}
\newtheorem{proposition}[theorem]{Proposition}
\newtheorem{corollary}[theorem]{Corollary}
\newtheorem{conjecture}{Conjecture}
\newtheorem*{conjecture*}{Conjecture}

\theoremstyle{definition}
\newtheorem{definition}[theorem]{Definition}
\newtheorem{remark}{Remark}

\newcounter{example}
\newenvironment{example}{\refstepcounter{example}\par\medskip\noindent\textbf{Example~\theexample.}}{\par\hfill$\triangle$\par\medskip}
\newenvironment{example*}{\par\medskip\noindent\textbf{Example.}}{\par\hfill$\triangle$\par\medskip}

\newcommand{\bra}[1]{\ensuremath \langle{#1}|}%
\newcommand{\ket}[1]{{\ensuremath |{#1}\rangle}}%
\newcommand{\bracket}[2]{\ensuremath \langle{#1}|{#2}\rangle}%
\newcommand{\h}[1]{\ensuremath \mathfrak{#1}}%
\newcommand{\indicator}[1]{\ensuremath \openone_{\{#1\}}}%
\newcommand{\id}{\ensuremath \mathrm{id}}%
\newcommand{\tr}{\ensuremath{\mathrm{tr}}}
\newcommand{\ip}[2]{\ensuremath \langle{#1}|{#2}\rangle}%
\newcommand{\category}[1]{\ensuremath{\mathsf{#1}}}
\newcommand{\Hom}{\mathrm{Hom}}

\newcommand{\covec}[1]{\ensuremath \undertilde{#1}}

\newcommand{\bcl}[1]{{\color{orange}\textbf{(Brad: #1)}}}
\newcommand{\ngr}[1]{{\color{magenta}\textbf{(Nishant: #1)}}}

\begin{abstract}
Nonlocal games with synchronous correlations are a natural generalization of functions between two finite sets. In this work we examine analogues of Bell's inequalities for such correlations, and derive a synchronous device-independent quantum key distribution protocol. This protocol has the advantage of symmetry between the two users and self-testing while generating shared secret key without requiring a preshared secret. We show that, unlike general correlations and the CHSH inequality, there can be no quantum Bell violation among synchronous correlations with two measurement settings. However we exhibit explicit analogues of Bell's inequalities for synchronous correlations with three measurement settings and two outputs, provide an analogue of Tsirl'son's bound in this setting, and prove existence and rigidity of quantum correlations that saturate this bound. We conclude by posing a security assumption that bypasses the locality, or causality, loophole and examine the protocol's robustness against measurement error and depolarization noise.
\end{abstract}

\maketitle

\begin{section}{Introduction}

In this work we describe a device-independent quantum key distribution protocol that is symmetric between Alice and Bob using the notion of a \emph{synchronous} correlation. Namely we study two player nonlocal games, which are characterized by a conditional probability distribution $p(y_A,y_B \:|\: x_A,x_B)$ where $x_A,x_B \in X$ and $y_A,y_B \in Y$. This is not the most general case we could consider, where each of Alice's and Bob's inputs and outputs could come from different sets. However, as our goal is to have symmetry under exchange of Alice and Bob and so such generality is unneeded.

At a high-level each round the protocol operates as follows:

\begin{center}
\fbox{
    \begin{minipage}{0.8\textwidth}
        \begin{enumerate}
            \item Alice and Bob share an EPR pair $\frac{1}{\sqrt{2}}(\ket{00} + \ket{11})$.
            \item Independently, Alice and Bob randomly select from one of three fixed measurement bases to measure his or her half of the EPR pair.
            \item After measurement, Alice and Bob exchange the basis selection they made.
            \item If they selected the same basis, they store the output of their measurements as a shared secret value. If they chose differing bases, they exchange their measurement outcomes and store these for later performing a self-test of the device.
        \end{enumerate}
    \end{minipage}
}
\end{center}

Note that Alice and Bob are selecting from the same set of measurement bases and their roles are completely exchangeable (compare to \cite{barrett2005nonlocal, vazirani2014fully}). Moreover, each round of the protocol uses these same set of bases and so there is no need for a preshared secret for heralding testing, or game, rounds \cite{pironio2010random, coudron2013robust, miller2016robust}. Yet, when Alice and Bob choose the same bases the protocol assumes that their measurement outcomes are perfectly correlated. This requires that the device be modeled as a ``synchronous'' correlation as defined here.

\begin{definition}
A correlation is \emph{synchronous} if
\begin{equation}\label{eqn:syncrhonous:definition}
   p(y_A,y_B\:|\:x,x) = 0 \text{ if $x\in X$ and $y_A\not= y_B$ $\in Y$.} 
\end{equation}
A correlation is \emph{symmetric} if $p(y_A,y_B\:|\:x_A,x_B) = p(y_B,y_A\:|\:x_B,x_A)$.
\end{definition}

As is traditional with schemes such the CHSH or Magic Square games, or their generalizations \cite{mermin1990simple, peres1990incompatible, cleve2004consequences, arkhipov2012extending, coladangelo2017robust}, the analysis relies on understanding the space of local (or ``classical'' or ``hidden variables'') correlations. Formally these are defined as follows.

\begin{definition}
    A \emph{local hidden variables strategy}, or simply \emph{classical correlation}, is a correlation of the form
    \begin{equation*}
        p(y_A,y_B\:|\:x_A,x_B) = \sum_{\omega\in\Omega} \mu(\omega) p_A(y_A\:|\: x_A,\omega) p_B(y_B\:|\: x_B,\omega)
    \end{equation*}
    for some finite set $\Omega$ and probability distribution $\mu$. A \emph{quantum correlation} is a correlation that takes the form
    \begin{equation*}
        p(y_A,y_B\:|\:x_A,x_B) = \tr(\rho(E^{x_A}_{y_A}\otimes F^{x_B}_{y_B}))
    \end{equation*}
    where $\rho$ is a density operator on the Hilbert space $\h{H}_A\otimes\h{H}_B$, and for each $x\in X$ we have $\{E^x_y\}_{y\in Y}$ and $\{F^x_y\}_{y\in Y}$ are POVMs on $\h{H}_A$ and $\h{H}_B$ respectively. 
\end{definition}

Local, or ``classical'' or ``hidden variables,'' synchronous correlations arise from simple strategies. Alice and Bob initially agree on a function $f:X \to Y$ and then when asked questions $x_A$ (resp. $x_B$) answer with $y_A = f(x_A)$ (resp. $y_B = f(x_B)$). In fact, these functions $f:X\to Y$ used as above are precisely the extreme points of the classical synchronous correlations. Consequently, classical synchronous correlations are always symmetric. An analogous result for quantum correlations, states that every quantum synchronous correlation is a convex combination of tracial states \cite{paulsen2016estimating}. So the symmetry of our protocol above was not circumstance: every quantum synchronous correlations must be symmetric. Formally we state these results here for later use; proofs of them can be found in the literature in various contexts \cite{cameron2007quantum, atserias2016quantum, manvcinska2016quantum, paulsen2016estimating}. For completeness, the reader can find direct proofs in Appendix
\ref{appendix:extreme-points}.

\begin{theorem}\label{theorem:classical:characterization}
The set of synchronous classical correlations with input $X$ and output $Y$ is bijective to the collection of probability distributions on the set of functions $X \to Y$. Given such a probability distribution, the associated strategy is: Alice and Bob sample a function $f:X\to Y$ according the specified distribution, and upon receiving $x_A,x_B$ they output $y_A = f(x_A)$ and $y_B = f(x_B)$.
\end{theorem}

\begin{corollary}\label{corollary:classical-extreme}
    The extreme points of the synchronous hidden variables strategies from $X$ to $Y$ can be canonically identified with the set of functions $X \to Y$.
\end{corollary}

\begin{corollary}\label{corollary:classical:symmetric}
    Every synchronous classical strategy is symmetric.
\end{corollary}

\begin{lemma}
    Let $p(y_A,y_B\:|\:x_A,x_B) = \tr(\rho(E^{x_A}_{y_A}\otimes F^{x_B}_{y_B}))$ be a synchronous quantum correlation. Then the POVMs $\{E^x_y\}_{y\in Y}$ and $\{F^x_y\}_{y\in Y}$, for $x\in X$, are projection-valued measures. Moreover each $E^x_y$ commutes with $\tr_B(\rho)$ and each $F^x_y$ commutes with $\tr_A(\rho)$.
\end{lemma}

\begin{lemma}
    Every synchronous quantum correlation can be expressed as the convex combination of synchronous quantum correlations with maximally entangled pure states. In particular, if a synchronous quantum correlation $\tr(\rho(E^{x_A}_{y_A}\otimes F^{x_B}_{y_B}))$ is extremal then we may take $\rho = \ket\psi\bra\psi$ with $\ket\psi$ maximally entangled.
\end{lemma}

\begin{theorem}\label{theorem:tracial-state}
Let $X,Y$ be finite sets, $\h{H}$ a $d$-dimensional Hilbert space, and for each $x\in X$ a projection-valued measure $\{E^x_y\}_{y\in Y}$ on $\h{H}$. Then
$$p(y_A,y_B\:|\: x_A,x_B) = \frac{1}{d}\tr(E^{x_A}_{y_A}E^{x_B}_{y_B})$$
defines a synchronous quantum correlation. Moreover every synchronous quantum correlation with maximally entangled pure state has this form.
\end{theorem}

\begin{corollary}\label{corollary:quantum-symmetry}
    Every synchronous quantum correlation is symmetric.
\end{corollary}

In complete generality, correlations as given above allow arbitrary communication between Alice and Bob. Primarily we focus on \emph{nonlocal games}, by which we mean Alice and Bob may utilize preshared information (such as the EPR pair in the protocol) but cannot communicate once they receive their inputs $x_A,x_B$ \cite{cleve2004consequences}. Recall the well-known nonsignaling conditions \cite{popescu1994quantum}.

\begin{definition}
    A correlation $p$ is \emph{nonsignaling} if it satisfies (i) for all $y_A,x_A,x_B,x_B'$
    $$\sum_{y_B} p(y_A,y_B\:|\:x_A,x_B) = \sum_{y_B} p(y_A,y_B\:|\:x_A,x_B'),$$
    and (ii) for all $y_B,x_B,x_A,x_A'$
    $$\sum_{y_A} p(y_A,y_B\:|\:x_A,x_B) = \sum_{y_A} p(y_A,y_B\:|\:x_A',x_B).$$
\end{definition}
All classical and quantum correlations are nonsignalling. In fact, the nonsignaling correlations form a polytope with the classical correlations a subpolytope of it \cite{tsirelson1993some,van2013implausible}. A Bell inequality is a facet of the classical polytope that is not a facet of the nonsignaling polytope. In the traditional Bell inequality, from which standard device-independent quantum key distribution protocols are derived, Alice and Bob are each asked two questions and each provide two answers \cite{bell1964einstein,clauser1969proposed}. However, that correlation is neither symmetric nor synchronous. In fact, in Appendix \ref{appendix:case2-2} we prove that any synchronous correlation where Alice and Bob have two measurement bases (regardless of the number of measurement outcomes) has no Bell inequalities in the sense that any symmetric synchronous nonsignaling correlation is classical. In particular, any synchronous quantum correlation with only two bases must be classical. 

Starting in the next section we will characterize the symmetric and synchronous nonsignaling correlation with two outcomes (regardless of the number of bases), and find that in the case when $|X| = 3$ and $|Y| = 2$ there are precisely four synchronous Bell inequalities. Some similar results have been found for asynchronous correlations \cite{jones2005interconversion,collins2004relevant}. We then examine violation of these Bell inequalities and prove a synchronous analogue of Tsirel'son bound: a synchronous quantum correlation can violate one such inequality by at most $\frac{1}{8}$. We prove rigidity of the correlations that achieve maximal violations. Any of these isolated correlations can be used as a self-test for an EPR pair \cite{mayers1998quantum, mayers2003self, mckague2012robust}. With the proper choice of three measurement bases, our protocol can implement such a self-test while simultaneously generating shared keys.

Unfortunately, this introduces a ``synchronicity loophole'' in that the rigidity of these quantum correlations are restricted to the synchronous case. In Section \ref{section:measures}, we extend beyond synchronous correlations and show that there are natural measures of asymmetry, a form of bias, and asychronicity, and these bound the potential synchronous Bell violations realizable by general classical correlations.

Like with CHSH inequalities, the correlations we produce can easily be simulated with sufficient classical communication \cite{brassard1999cost, toner2003communication}, leading to the ``locality'' or ``causality'' loopholes common to device-independent protocols. In Section \ref{section:causality-loophole}, we provide an alternative security assumption that allows us to circumvent this: if Eve has even a small uncertainty about of the choice of bases Alice and Bob use to measure, then she has limited ability to generate a Bell violation regardless of her communication or computing power. Our method of analysis also allows us to examine the robustness of the protocol under several noise models.

\end{section}

\begin{section}{Synchronous correlations}\label{section:synchronous}

When studying correlations with $|X| = n$ and $|Y|=2$, and for concreteness say $Y = \{0,1\}$, it is particularly fruitful to work with the traditional biases and correlation matrices:
\begin{align*}
    a_{x_A} &= \sum_{y_A, y_B}{(-1)^{(1,0)\cdot(y_A, y_B)} p(y_A, y_B | x_A, x_B)}\\
    b_{x_B} &= \sum_{y_A, y_B}{(-1)^{(0,1)\cdot(y_A, y_B)} p(y_A, y_B | x_A, x_B)}\\
    c_{x_A,x_B} &= \sum_{y_A, y_B}{(-1)^{(1,1)\cdot(y_A, y_B)} p(y_A, y_B | x_A, x_B)}.
\end{align*}
Note that the nonsignalling criteria implies that $a$ and $b$ do not depend on $x_B$ or $x_A$ respectively. Any probability distribution has
$$1 = \sum_{y_A, y_B}{(-1)^{(0,0)\cdot (y_A, y_B)} p(y_A, y_B | x_A, x_B)}$$
and so we can invert these as
\begin{equation}\label{eqn:inverse-relation}
\begin{array}{rl}
    p(0,0| x_A, x_B) &= \tfrac{1}{4}\left(1 + a_{x_A} + b_{x_B} + c_{x_A, x_B}\right),\\
    p(0,1| x_A, x_B) &= \tfrac{1}{4}\left(1 + a_{x_A} - b_{x_B} - c_{x_A, x_B}\right),\\
    p(1,0| x_A, x_B) &= \tfrac{1}{4}\left(1 - a_{x_A} + b_{x_B} - c_{x_A, x_B}\right),\\
    p(1,1| x_A, x_B) &= \tfrac{1}{4}\left(1 - a_{x_A} - b_{x_B} + c_{x_A, x_B}\right).
\end{array}
\end{equation}
As each $p$ term is nonnegative we have the $4n^2$ inequalities in $n^2+2n$ variables
\begin{equation}\label{eqn:nonsignaling-inequalities}
\begin{array}{rl}
    1 + a_{x_A} + b_{x_B} + c_{x_A, x_B} &\geq 0,\\
    1 + a_{x_A} - b_{x_B} - c_{x_A, x_B} &\geq 0,\\
    1 - a_{x_A} + b_{x_B} - c_{x_A, x_B} &\geq 0,\\
    1 - a_{x_A} - b_{x_B} + c_{x_A, x_B} &\geq 0,
\end{array}
\end{equation}    
which form the basis for our description of the polytope of nonsignaling correlations below.

\begin{lemma}\label{lemma:symmetric-reduction}
    A correlation $p$ is symmetric and nonsignalling if and only if (i) $c_{x_A,x_B} = c_{x_B,x_A}$ and (ii) $a_x = b_x$.
\end{lemma}
\begin{proof}
    Suppose $p$ is symmetric. Then 
    \begin{align*}
        a_{x_A} &= \sum_{y_A, y_B}{(-1)^{(1,0)\cdot(y_A, y_B)} p(y_A, y_B | x_A, x_B)}\ 
        =\ \sum_{y_A, y_B}{(-1)^{(1,0)\cdot(y_A, y_B)} p(y_B, y_A | x_B, x_A)}\\
        &= \sum_{y_A, y_B}{(-1)^{(1,0)\cdot(y_B, y_A)} p(y_A, y_B | x_B, x_A)}\ 
        =\ \sum_{y_A, y_B}{(-1)^{(0,1)\cdot(y_A, y_B)} p(y_A, y_B | x_B, x_A)}\ 
        =\ b_{x_A}.
    \end{align*}
    Similarly,
    \begin{align*}
        c_{x_A, x_B} &= \sum_{y_A, y_B}{(-1)^{(1,1)\cdot(y_A, y_B)} p(y_A, y_B | x_A, x_B)}\ 
        =\ \sum_{y_A, y_B}{(-1)^{(1,1)\cdot(y_A, y_B)} p(y_B, y_A | x_B, x_A)}\\
        &= \sum_{y_A, y_B}{(-1)^{(1,1)\cdot(y_B, y_A)} p(y_A, y_B | x_B, x_A)}\ 
        =\ \sum_{y_A, y_B}{(-1)^{(1,1)\cdot(y_A, y_B)} p(y_A, y_B | x_B, x_A)}\ 
        =\ c_{x_B,x_A}.
    \end{align*}
\end{proof}

From this lemma, the polytope of symmetric nonsignalling correlations lives in a $\frac{1}{2}n(n-1)$ dimensional space. Its facets are just the inequalities (\ref{eqn:nonsignaling-inequalities}) reduced by the equations in this lemma. 
\begin{equation}\label{eqn:symmetric-nonsignaling-inequalities}
\begin{array}{rl}
    \left.\begin{aligned}
        1 - c_{j,j} &\geq 0\\
        1 + 2a_j + c_{j,j} &\geq 0\\
        1 - 2a_j + c_{j,j} &\geq 0
    \end{aligned}\right\}
    & \text{ for $j = 0, \dots, n-1$}\\
    \left.\begin{aligned}
        1 + a_j + a_k + c_{j,k} &\geq 0\\
        1 + a_j - a_k + c_{j,k} &\geq 0\\
        1 - a_j + a_k + c_{j,k} &\geq 0\\
        1 - a_j - a_k + c_{j,k} &\geq 0
    \end{aligned}\right\}
    & \text{ for $0 \leq j < k \leq n-1$. }
\end{array}
\end{equation}
We will study this polytope more in \S{\ref{section:measures}}.

\begin{lemma}\label{lemma:synchonous-reduction}
    A correlation $p$ is synchronous and nonsignalling if and only if for all $x\in X$ we have (i) $c_{x,x} = 1$ and (ii) $a_x = b_x$.
\end{lemma}
\begin{proof}
    Assume $p$ is synchronous and nonsignalling. Thus $p(y_A, y_B | x, x) = 0$ if $y_A \neq y_B$. We always have
    \begin{equation}\label{eqn:synch1} 
        p(0,0| x, x) + p(0,1| x, x)+ p(1,0| x, x) + p(1,1| x, x) = p(0,0| x, x) + p(1,1| x, x) = 1.
    \end{equation}
    We have:
    \begin{align*}
        c_{x,x} &= \sum_{y_A, y_B}{(-1)^{(1,1).(y_A, y_B)} p(y_A, y_B | x, x)}\\
        &= p(0,0| x, x) - p(0,1| x, x) - p(1,0| x, x) + p(1,1| x, x)\\
        &= p(0,0| x, x) + p(1,1| x, x) = 1 \quad \text{from (\ref{eqn:synch1}).}
    \end{align*}
    Similarly,
    \begin{align*}
        a_x &= \sum_{y_A, y_B}{(-1)^{(1,0).(y_A, y_B)} p(y_A, y_B | x, x)}\\
        &= p(0,0| x, x) + p(0,1| x, x) - p(1,0| x, x) - p(1,1| x, x) = p(0,0| x, x) - p(1,1| x,x)\\
        &= p(0,0| x, x) - p(0,1| x, x) + p(1,0| x, x) - p(1,1| x, x) = \sum_{y_A, y_B}{(-1)^{(0,1).(y_A, y_B)} p(y_A, y_B | x, x)} = b_x
    \end{align*}
    
    Now we prove the other direction. Assume $c_{x,x} =1$, and $a_x = b_x$. For $p$ to be synchronous we need $p(0,1|x,x) = 0 = p(1,0|x,x)$.
    \begin{align*}
        p(0,1|x,x) &= 1 + a_{x} - b_{x} - c_{x,x} = 1 + a_{x} - a_{x} - 1 = 0, \text{ and}\\
        p(1,0|x,x) &= 1 - a_{x} + b_{x} - c_{x,x} = 1 - a_{x} + a_{x} - 1 = 0.\\
    \end{align*}
    Thus, we have $p$ is a synchronous correlation.
\end{proof}

Similarly, the equations of this lemma imply the polytope of synchronous nonsignalling correlations live in an $n^2$ dimensional space. Its $n^2$ facets are defined by the $4n(n-1)$ inequalities:
\begin{equation}\label{eqn:synchronous-nonsignaling-inequalities}
\begin{array}{rl}
    \left.\begin{aligned}
        1 + a_j + a_k + c_{j,k} &\geq 0\\
        1 + a_j - a_k - c_{j,k} &\geq 0\\
        1 - a_j + a_k - c_{j,k} &\geq 0\\
        1 - a_j - a_k + c_{j,k} &\geq 0
    \end{aligned}\right\}
    & \text{ for $0 \leq j \not= k \leq n-1$. }
\end{array}
\end{equation}
Note that a reduction of the inequalities (\ref{eqn:nonsignaling-inequalities}) by the conditions in the lemma also provide that $-1 \leq a_j \leq 1$ (for $j=0, \dots, n-1$), but these follow from (\ref{eqn:synchronous-nonsignaling-inequalities}).

Yet from Corollary \ref{corollary:classical:symmetric}, all synchronous hidden variables correlations are symmetric. So we are primarily interested in the polytope of symmetric and synchronous nonsignaling correlations. This $\frac{n(n+1)}{2}$-dimensional polytope has $2n(n-1)$ facets given by:
\begin{equation}\label{eqn:symmetric-synchronous-nonsignaling-inequalities}
\begin{array}{rl}
    \left.\begin{aligned}
        1 + a_j + a_k + c_{j,k} &\geq 0\\
        1 + a_j - a_k - c_{j,k} &\geq 0\\
        1 - a_j + a_k - c_{j,k} &\geq 0\\
        1 - a_j - a_k + c_{j,k} &\geq 0
    \end{aligned}\right\}
    & \text{ for $0 \leq j < k \leq n-1$. }
\end{array}
\end{equation}
We will also study this polytope more in \S{\ref{section:measures}}.

To derive synchronous Bell inequalities, we must examine the structure of the local hidden variables polytope. From Corollary \ref{corollary:classical-extreme}, this polytope of correlations has the functions $f:X \to Y$ as its vertices. For $Y=\{0,1\}$ and $X = \{0,\dots,k\}$ we can associate each such function to an integer in $r \in \{0,\dots, 2^k-1\}$: the function $f_r$ has $f_r(x) = j$ where $j$ is the $x$-th bit of $r$. In particular, for $X = \{0,1,2\}$ there are eight functions $\{f_0, \dots, f_7\}$. We will use this representation in some concrete examples below.

For general $X$ this polytope is quite complicated, but for $X = \{0,1,2\}$ it is easy to describe: its facets are the twelve facets (\ref{eqn:symmetric-synchronous-nonsignaling-inequalities}) above plus the additional ones:
\begin{equation}\label{eqn:Bell-inequalities}
    \begin{array}{rll}
    J_0 &= \tfrac{1}{4}\left(1 - c_{01} - c_{02} + c_{12}\right) &\geq 0\\
    J_1 &= \tfrac{1}{4}\left(1 - c_{01} + c_{02} - c_{12}\right) &\geq 0\\
    J_2 &= \tfrac{1}{4}\left(1 + c_{01} - c_{02} - c_{12}\right) &\geq 0\\
    J_3 &= \tfrac{1}{4}\left(1 + c_{01} + c_{02} + c_{12}\right) &\geq 0.
    \end{array}
\end{equation}
So a symmetric synchronous nonsignaling correlation, which includes any synchronous quantum correlation, is classical if and only if these four inequalities are satisfied. We refer to these as \emph{synchronous Bell inequalities} for this reason.

Turning to synchronous quantum correlations, according to Tsirl'son's theorem \cite{tsirelson1980quantum} if $p$ is a quantum correlation then there exist unit vectors $\{\vec{u}_x,\vec{v}_x\}_{x \in X}$ in some real inner product spaces such that $c_{x_A,x_B} = \langle \vec{u}_{x_A}, \vec{v}_{x_B}\rangle$. Conversely, given a matrix $C = (c_{x_A,x_B})$ with $c_{x_A,x_B} = \langle \vec{u}_{x_A}, \vec{v}_{x_B}\rangle$ there exists a quantum correlation with this correlation matrix (in particular, one with $a_x = b_x = 0$). 

\begin{lemma}\label{lemma:quantum-reduction}
    A correlation matrix $C$ is synchronous and quantum if and only if there exists unit vectors $\{\vec{u}_x\}$ such that $c_{x_A,x_B} = \langle \vec{u}_{x_A}, \vec{u}_{x_B}\rangle$.
\end{lemma}
\begin{proof}
    We assume the correlation is quantum synchronous, and thus have that $c_{x,x} = \langle \vec{u}_x, \vec{v}_x \rangle = 1$ for all $x \in \{0. \dots, n-1\}$. By Cauchy-Schwarz we have $1 = \langle \vec{u}_x, \vec{v}_x \rangle \leq ||\vec{u}_x|| ||\vec{v}_x|| \leq 1$. Therefore, since equality holds in the Cauchy-Schwarz inequality, we must have $\vec{u}_x = \vec{v}_x$ for all $x$.

    Conversely, assume $\vec{u}_x = \vec{v}_x$ for all $x \in \{0, \dots, n-1\}$. Thus for all $x$, we have $c_{x,x} = \langle \vec{u}_{x}, \vec{v}_{x} \rangle = \langle \vec{u}_{x}, \vec{u}_{x} \rangle = ||\vec{u}_x||^2 = 1$. Thus, $C$ is quantum synchronous.
\end{proof}

Note the similarity of this result with the representation of a quantum correlation by tracial states, Theorem \ref{theorem:tracial-state}. So Lemmas \ref{lemma:symmetric-reduction}, \ref{lemma:synchonous-reduction}, and \ref{lemma:quantum-reduction} provide an alternative proof of Corollary \ref{corollary:quantum-symmetry}, that every quantum synchronous correlation must be symmetric.

\end{section}

\begin{section}{Tsirl'son bounds and rigidity}

In the previous section, we characterized the polytopes of synchronous symmetric nonsignaling and hidden variables strategies from $X = \{0,1,2\}$ to $Y = \{0,1\}$, and obtained four synchronous analogous of Bell inequalities (\ref{eqn:Bell-inequalities}). In this section we give concrete strategies that violate these inequalities. But first note that there are relationships between them. For example summing the definitions of the four $J$-terms in (\ref{eqn:Bell-inequalities}) gives $J_0 + J_1 + J_2 + J_3 = 1$.
In particular, we note that for $0 \leq j < k \leq 2$ we have
$$J_j + J_k = \tfrac{1}{2}(1 - c_{j,k}).$$
Even in the full polytope of all nonsignaling correlations we have $1 - c_{j,k} \geq 0$, obtained by summing the middle two inequalities of (\ref{eqn:nonsignaling-inequalities}). A similar argument works when $J_3$ is one of the terms; for example
$$J_0 + J_3 = 1 - J_1 - J_2 = \frac{1}{2}(1 + c_{j,k}) \geq 0$$
from summing the outer two inequalities of (\ref{eqn:nonsignaling-inequalities}). Therefore, if a (necessarily nonclassical) nonsignaling correlation violates a Bell inequality, say $J_3 < 0$, then it must have $J_1,J_2,J_3 > 0$ as pairwise these sum to a nonnegative value.

\begin{proposition}
    Every synchronous symmetric nonsignaling strategy satisfies $J_0, J_1, J_2, J_3 \geq -\frac{1}{2}$. However no individual correlation can violate more than one of the inequalities $J_0, J_1, J_2, J_3\geq 0$.
\end{proposition}
\begin{proof}
    We obtain this bound by summing the appropriate inequalities. For instance,
    \begin{align*}
        J_0 + J_3 &\geq 0\\
        J_1 + J_3 &\geq 0\\
        J_2 + J_3 &\geq 0
    \end{align*}
    and using $J_0 + J_1 + J_2 + J_3 = 1$ produces $1 - 2J_3 \geq 0$.
\end{proof}

Quantum correlations can also violate the Bell inequalities (\ref{eqn:Bell-inequalities}). And these will also have similar bounds, which  typically called Tsirl'son bounds after his seminal work on Bell's inequality \cite{tsirelson1980quantum}. As quantum correlations are nonsignaling, the above argument shows that only one inequality can be violated for any given correlation. Our analogue of Tsirl'son's bound is given in the following theorem. The next theorem proves that these bounds are saturated by a unique quantum correlation.

\begin{theorem}\label{theorem:synchronous-quantum-bounds}
    Every synchronous quantum correlation satisfies $J_0, J_1, J_2, J_3 \geq -\frac{1}{8}$. However no individual correlation can violate more than one of the inequalities $J_0, J_1, J_2, J_3 \geq 0$.
\end{theorem}
\begin{proof}
    Since quantum correlations are nonsignaling, the second claim follows as indicated above. To prove the stated bounds, we note the space of quantum strategies is convex and each $J_j$ is linear, thus the extreme values will occur at extremal quantum strategies. By Theorem \ref{theorem:tracial-state} these take the form
    $$p(y_A,y_B\:|\: x_A,x_B) = \frac{1}{d}\tr(E^{x_A}_{y_A}E^{x_B}_{y_B}),$$
    for projection-valued measures $\{E^x_0,E^x_1\}_{x = 0,1,2}$ on $\h{H} = \mathbb{C}^d$. Recall that
    \begin{align*}
        a_x &= p(0,0\:|\:x,x) - p(1,1\:|\:x,x),\text{ and}\\
        c_{x_Ax_B} &= p(0,0\:|\:x_A,x_B) + p(1,1\:|\:x_A,x_B) - p(0,1\:|\:x_A,x_B) - p(1,0\:|\:x_A,x_B),
    \end{align*}
    and so by defining the $\pm 1$-valued observables $M_x = E^x_0 - E^x_1$, we have
    $$a_x = \frac{1}{d}\tr(M_x) \text{ and } c_{x_Ax_B} = \frac{1}{d}\tr(M_{x_A}M_{x_B}).$$
    So focussing on say $J_3$ we compute
    \begin{align*}
        \frac{1}{d}\tr((M_0 + M_1 + M_2)^2) &= \frac{1}{d}\left[\tr(M_0^2) + \tr(M_1^2) + \tr(M_2^2)\right.\\
        &\quad + \left.2\tr(M_0M_1) + 2\tr(M_0M_2) + 2\tr(M_1M_2)\right]\\
        &= 3 + 2(c_{0,1} + c_{0,2} + c_{1,2}) = 1 + 8J_3.
    \end{align*}
    Therefore
    $$J_3 = -\frac{1}{8} + \frac{1}{8d}\tr((M_0 + M_1 + M_2)^2) \geq -\frac{1}{8}.$$
    
    Similarly bounds on $J_0$, $J_1$, and $J_2$ arise from $\tr((-M_0 + M_1 + M_2)^2)$, $\tr((M_0 - M_1 + M_2)^2)$, and $\tr((M_0 + M_1 - M_2)^2)$ respectively.
\end{proof}

\begin{theorem}\label{thm:rigidity}
    For each of the bounds of Theorem \ref{theorem:synchronous-quantum-bounds}, there exists a unique synchronous quantum correlation from $\{0,1,2\}$ to $\{0,1\}$ that achieves it.
\end{theorem}
\begin{proof}
    We prove this for $J_3 = -\frac{1}{8}$; the cases for $J_0, J_1, J_2$ are similar.
    
    Continuing the notation above, suppose we have observables $M_0, M_1, M_2$ on a Hilbert space $\mathfrak{H}$ for which
    $$-\frac{1}{8} = J_3 = \frac{1}{4}\left(1 + \frac{1}{d}\tr(M_0M_1) + \frac{1}{d}\tr(M_0M_2) + \frac{1}{d}\tr(M_1M_2)\right).$$
    This implies $\frac{1}{d}\tr((M_0 + M_1 + M_2)^2) = 1 + 8J_3 = 0$, and hence $M_0 + M_1 + M_2 = 0$. In particular, taking the square of $M_2 = -(M_0 + M_1)$ gives
    $$\openone = (M_0 + M_1)^2 = 2\openone + M_0M_1 + M_1M_0.$$
    That is, the anticommutator $\{M_0,M_1\} = -\openone$.

    We apply two projections theory \cite{halmos1969two, amrein1994pairs, boettcher2010gentle}, which is the predecessor to Jordan's principle component decomposition, to $E^0_0$ and $E^1_0$. This provides a decomposition of our Hilbert space
    $$\mathfrak{H} = \mathfrak{L}_{00}\oplus \mathfrak{L}_{01} \oplus \mathfrak{L}_{10} \oplus \mathfrak{L}_{11} \oplus \bigoplus_{j=1}^k \mathfrak{H}_j$$
    where $\mathfrak{L}_{\lambda\mu}$ is a joint eigenspace of $E^0_0$ and $E^1_0$ associated to eigenvalues $\lambda$ and $\mu$ respectively, and on the two-dimensional subspace $\mathfrak{H}_j$ there exists a basis for which the action of $E^0_0$ and $E^1_0$ have matrix representation
    $$[E^0_0] = \left(\begin{array}{cc} 1 & 0 \\ 0 & 0\end{array}\right), \text{ and } [E^1_0] = \left(\begin{array}{cc} \cos^2\theta_j & \sin\theta_j\cos\theta_j \\ \sin\theta_j\cos\theta_j & \sin^2\theta_j\end{array}\right).$$
    If $\dim\mathfrak{L}_{\lambda\mu} > 0$ then
    $$\left.M_0\right|_{\mathfrak{L}_{\lambda\mu}} = (-1)^\lambda\openone|_{\mathfrak{L}_{\lambda\mu}} \text{ and } \left.M_1\right|_{\mathfrak{L}_{\lambda\mu}} = (-1)^\mu \openone|_{\mathfrak{L}_{\lambda\mu}},$$
    however this implies 
    $$\left.M_2^2\right|_{\mathfrak{L}_{\lambda\mu}} = ((-1)^\lambda + (-1)^\mu)^2\openone|_{\mathfrak{L}_{\lambda\mu}} \not= \openone|_{\mathfrak{L}_{\lambda\mu}}.$$
    Hence $\dim\mathfrak{L}_{\lambda\mu} = 0$ and $\mathfrak{H} = \bigoplus_{j=1}^{d/2} \mathfrak{H}_j$ where $d = \dim\mathfrak{H}$ must be even.

    Turning to $\mathfrak{H}_j$ we have
        $$[\left.M_0\right|_{\mathfrak{H}_j}] = \left(\begin{array}{cc} 1 & 0 \\ 0 & -1 \end{array}\right) \text{ and } [\left.M_1\right|_{\mathfrak{H}_j}] = \left(\begin{array}{cc} \cos^2\theta_j - \sin^2\theta_j & 2\sin\theta_j\cos\theta_j \\ 2\sin\theta_j\cos\theta_j & \sin^2\theta_j - \cos^2\theta_j\end{array}\right).$$
    We then compute
        $$[\{\left.M_0\right|_{\mathfrak{H}_j},\left.M_1\right|_{\mathfrak{H}_j}\}] = 2(\cos^2\theta_j - \sin^2\theta_j)\openone,$$
    and so conclude that $\theta_j = \frac{1}{2}\cos^{-1}\left(\frac{1}{2}\right)$. In particular, on every $\mathfrak{H}_j$
        \begin{equation}\label{eqn:J3measurement}
        [\left.M_0\right|_{\mathfrak{H}_j}] = \left(\begin{array}{cc} 1 & 0 \\ 0 & -1 \end{array}\right),\   [\left.M_1\right|_{\mathfrak{H}_j}] = \frac{1}{2}\left(\begin{array}{cc} -1 & \sqrt{3} \\ \sqrt{3} & 1\end{array}\right) \text{ and } [\left.M_2\right|_{\mathfrak{H}_j}] = \frac{1}{2}\left(\begin{array}{cc} -1 & -\sqrt{3} \\ -\sqrt{3} & 1\end{array}\right).
        \end{equation}
    In particular we have
    $$\begin{array}{rl@{\qquad}rl}
        c_{0,1} &= \tfrac{1}{d}\tr(M_0) = 0, & c_{0,1} &= \tfrac{1}{d}\tr(M_0M_1) = -\tfrac{1}{2},\\
        c_{0,2} &= \tfrac{1}{d}\tr(M_1) = 0, & c_{0,2} &= \tfrac{1}{d}\tr(M_0M_1) = -\tfrac{1}{2},\\
        c_{1,2} &= \tfrac{1}{d}\tr(M_2) = 0, & c_{1,2} &= \tfrac{1}{d}\tr(M_1M_2) = -\tfrac{1}{2}.
    \end{array}$$
    And therefore the matrix of the correlation must be
    \begin{equation}\label{eqn:J3correlation}
    [p(y_A,y_B|x_A,x_B)] = \frac{1}{8} \left(\begin{array}{ccccccccc} 
            4 & 1 & 1 & 1 & 4 & 1 & 1 & 1 & 4\\
            0 & 3 & 3 & 3 & 0 & 3 & 3 & 3 & 0\\
            0 & 3 & 3 & 3 & 0 & 3 & 3 & 3 & 0\\
            4 & 1 & 1 & 1 & 4 & 1 & 1 & 1 & 4
    \end{array}\right).
    \end{equation}
    
    Conversely, one easily computes that this matrix has $J_3 = -\frac{1}{8}$.
\end{proof}

The unique correlation (among synchronous quantum correlations) with $J_3 = -\frac{1}{8}$ has $a_0 = a_1 = a_2 = 0$ and $c_{0,1} = c_{0,2} = c_{1,2} = -\frac{1}{2}$. We could get even stronger violations among general synchronous symmetric nonsignaling correlations. For example $J_3 = -\frac{1}{2}$ for the correlation with $a_0 = a_1 = a_2 = 0$ and $c_{0,1} = c_{0,2} = c_{1,2} = -1$. Namely,
$$[p(y_A, y_B|x_A,x_B)] = \frac{1}{2}\left(\begin{array}{ccc|ccc|ccc}
1 & 0 & 0 & 0 & 1 & 0 & 0 & 0 & 1\\
0 & 1 & 1 & 1 & 0 & 1 & 1 & 1 & 0\\
0 & 1 & 1 & 1 & 0 & 1 & 1 & 1 & 0\\
1 & 0 & 0 & 0 & 1 & 0 & 0 & 0 & 1
\end{array}\right).$$
However to build such a correlation from nonsignaling components would require these have super-quantum behavior. 

\begin{example}
Recall that the PR-box is a two input/two output correlation given by
$$p(y'_A, y'_B|x'_A,x'_B) = \left\{\begin{array}{cl} \frac{1}{2} & \text{if $y'_A\oplus y'_B = x'_A\cdot x'_B$,} \\ 0 & \text{otherwise.} \end{array}\right.$$
That is we have the correlation matrix
$$[p(y'_A, y'_B|x'_A,x'_B)] = \frac{1}{2}\left(\begin{array}{cccc}
1 & 1 & 1 & 0 \\
0 & 0 & 0 & 1 \\
0 & 0 & 0 & 1 \\
1 & 1 & 1 & 0 
\end{array}\right).$$
We claim that Alice and Bob can simulate the previous correlation if they share three PR-boxes. Namely, each of Alice and Bob is given input $x_A,x_B\in\{0,1,2\}$; they evaluate each of the three PR-boxes as follows: for box $j$ Alice assigns input $x'_A = 0$ if $j = x_A$ and $x'_A = 1$ if $j \not= x_A$, and identically for Bob. Obtaining the output from each of the three boxes, Alice sets $y_A$ to be the parity of the results, and identically for Bob. Note that if Alice and Bob received the same input, $x_A = x_B$, they program the same inputs into their PR-boxes: two get $(1,1)$ and one gets $(0,0)$. Consequently, Alice and Bob receive outputs that differ in precisely two places, and hence obtain same parity. On the other hand, if Alice and Bob receive different inputs, $x_A \not= x_B$, then their inputs to the PR-boxes consist of one each of $(0,1)$, $(1,0)$, and $(1,1)$. In this case their outputs differ in precisely one place and therefore their parities disagree.
\end{example}

\end{section}

\begin{section}{Measures of asynchronicity and asymmetry}\label{section:measures}

The rigidity result of the previous section allow Alice and Bob to produce a certificate that they share maximally entangled states. In the context of the above theorem, if a device produces an entangled pair and they make measurements according to (\ref{eqn:J3measurement}), then they can achieve statistics according the associated correlation (\ref{eqn:J3correlation}) if and only if those shared states are maximally entangled. This forms crux of the security in our device-independent quantum key distribution protocol: when Alice and Bob choose the same basis they obtain a shared secret bit, when they choose different bases they build the certificate that the device is producing maximally entangled states.

The loophole in this argument is that this protocol is rigid among synchronous quantum protocols, and so in principle there may be asynchronous protocols that can achieve $J_3 = -\frac{1}{8}$ without using a maximally entangled state. In this section, and Appendix \ref{appendix:notes}, we show to how to treat asynchronicity and asymmetry.

Recall that the nonsignaling polytope is given by facets are given by (\ref{eqn:nonsignaling-inequalities}), where now $j,k=0,1,2$: 
\begin{equation}\label{eqn:nspoly-facets}
    \begin{array}{rl}
    1 + a_{j} + b_{k} + c_{j,k} &\geq 0,\\
    1 - a_{j} - b_{k} + c_{j,k} &\geq 0,\\
    1 - a_{j} + b_{k} - c_{j,k} &\geq 0,\\
    1 + a_{j} - b_{k} - c_{j,k} &\geq 0.
    \end{array}
\end{equation}
In particular, this defines a fifteen dimensional polytope in $\mathbb{R}^{15}$. In the asynchronous case, a classical hidden variables strategy can be generated by taking a random selection of a pair of functions $f_A,f_B:X \to Y$, where the Alice computes $y_A = f_A(x_A)$ and Bob computes $y_B = f_B(x_B)$. The subpolytope of classical correlations is therefore the convex hull of $64$ pairs $(f_j,f_k)$ where $j,k = 0,\cdots, 7$. Unfortunately this is not so tractible to analyze as it has $684$ facets (which can simply be generated by mathematical software such as \textsc{Sage} \cite{sagemath}).

In any case, we define a new sets of variables,
$$\begin{array}{rll}
    A_{j,k} &= \tfrac{1}{2}(c_{j,k} - c_{k,j}) & \text{(``asymmetry'')}\\
    B_j &= a_j - b_j & \text{(``bias'')}\\
    C_{j,k} &= \tfrac{1}{2}(c_{j,k} + c_{k,j}) & \text{(for $j \not= k$)}\\
    S_j &= 1 - c_{j,j} & \text{(``asyncrhonicity'')}\\
    U_j &= a_j + b_j.
\end{array}$$
Now it is straightforward to express $a_j,b_k,c_{j,k}$ in terms of these new coordinates:
\begin{align*}
    a_j &= \tfrac{1}{2}(U_j + B_j)\\
    b_j &= \tfrac{1}{2}(U_j - B_j)\\
    c_{j,j} &= 1 - S_j\\
    c_{j,k} &= C_{j,k} + A_{j,k} \text{ (when $j\not= k$).}
\end{align*}
In these coordinates, the nonsignaling polytope has facets in two families. When $j = k$ these reduce to:
\begin{align*}
    |B_j| &\leq S_j\\
    |U_j| &\leq 2 - S_j.
\end{align*}
For $j\not= k$ we have:
\begin{align*}
    1 + C_{j,k} + \tfrac{U_j + U_k}{2} + A_{j,k} + \tfrac{B_j - B_k}{2} &\geq 0\\
    1 + C_{j,k} - \tfrac{U_j + U_k}{2} + A_{j,k} - \tfrac{B_j - B_k}{2} &\geq 0\\
    1 - C_{j,k} - \tfrac{U_j - U_k}{2} - A_{j,k} - \tfrac{B_j + B_k}{2} &\geq 0\\
    1 - C_{j,k} + \tfrac{U_j - U_k}{2} - A_{j,k} + \tfrac{B_j + B_k}{2} &\geq 0.
\end{align*}

From the lemmas of Section \ref{section:synchronous} we have:
\begin{enumerate}
    \item A correlation is symmetric if and only if $A_{j,k} = 0$ and $B_j = 0$.
    \item A correlation is synchronous if and only if $S_j = 0$.
\end{enumerate}
Note that $S_j = 0$ implies that $B_j = 0$ from the inequalities above. Hence the polytope of synchronous symmetric nonsignaling correlations is given by $-2 \leq U_j \leq 2$ and
\begin{align*}
    1 + C_{j,k} + \tfrac{U_j + U_k}{2} &\geq 0\\
    1 + C_{j,k} - \tfrac{U_j + U_k}{2} &\geq 0\\
    1 - C_{j,k} - \tfrac{U_j - U_k}{2} &\geq 0\\
    1 - C_{j,k} + \tfrac{U_j - U_k}{2} &\geq 0,
\end{align*}
lying within the subspace given by $A_{j,k} = 0$, $B_j = 0$, and  $S_j = 0$. Within this subspace the hidden variables polytope satisfies the additional inequalities
\begin{align*}
    J_0 = \tfrac{1}{4}(1 - C_{0,1} - C_{0,2} + C_{1,2}) &\geq 0\\
    J_1 = \tfrac{1}{4}(1 - C_{0,1} + C_{0,2} - C_{1,2}) &\geq 0\\
    J_2 = \tfrac{1}{4}(1 + C_{0,1} - C_{0,2} - C_{1,2}) &\geq 0\\
    J_3 = \tfrac{1}{4}(1 + C_{0,1} + C_{0,2} + C_{1,2}) &\geq 0.
\end{align*}
The challenge is to characterize how these inequalities behave as we move away from the subspace defined by symmetry and synchronicity.

While the general nonsignaling polytope is easy to describe, namely through the 36 facets (\ref{eqn:nspoly-facets}), the full hidden variables polytope is not. In particular, the single facet $J_3 \geq 0$ of the synchronous hidden variables polytope arises from the intersection of any of twelve different facets of the full hidden variables polytope with the space $A = B = S = 0$. A similar fact holds for each of $J_0, J_1, J_2$, only accounting for 48 of the 684 facets of the hidden variables polytope. Nonetheless, the 12 facets covering $J_3 \geq 0$ are 
\begin{equation}\label{eqn:hv-inequalities}
    \begin{array}{rl}
        \begin{aligned}
            (1 + C_{0,1} + C_{0,2} + C_{1,2}) + S_0 - A_{0,1} + A_{0,2} + A_{1,2} &\geq 0,\\
            (1 + C_{0,1} + C_{0,2} + C_{1,2}) + S_0 + A_{0,1} - A_{0,2} - A_{1,2} &\geq 0,\\
            (1 + C_{0,1} + C_{0,2} + C_{1,2}) + S_1 + A_{0,1} + A_{0,2} + A_{1,2} &\geq 0,\\
            (1 + C_{0,1} + C_{0,2} + C_{1,2}) + S_2 + A_{0,1} + A_{0,2} - A_{1,2} &\geq 0,\\
            (1 + C_{0,1} + C_{0,2} + C_{1,2}) + S_2 - A_{0,1} - A_{0,2} + A_{1,2} &\geq 0,\\
            (1 + C_{0,1} + C_{0,2} + C_{1,2}) + S_1 - A_{0,1} - A_{0,2} - A_{1,2} &\geq 0,\\
            2(1 + C_{0,1} + C_{0,2} + C_{1,2}) - B_0 + B_1 + S_0 + S_1 &\geq 0,\\
            2(1 + C_{0,1} + C_{0,2} + C_{1,2}) + B_0 - B_1 + S_0 + S_1 &\geq 0,\\
            2(1 + C_{0,1} + C_{0,2} + C_{1,2}) + B_0 - B_2 + S_0 + S_2 &\geq 0,\\
            2(1 + C_{0,1} + C_{0,2} + C_{1,2}) - B_0 + B_2 + S_0 + S_2 &\geq 0,\\
            2(1 + C_{0,1} + C_{0,2} + C_{1,2}) + B_1 - B_2 + S_1 + S_2 &\geq 0,\\
            2(1 + C_{0,1} + C_{0,2} + C_{1,2}) - B_1 + B_2 + S_1 + S_2 &\geq 0.
        \end{aligned}
    \end{array}
\end{equation}

And in particular we find
\begin{align*}
    J_3 &\geq \frac{1}{4}\max\{-S_0 + |A_{0,1} - A_{0,2} - A_{1,2}|, -S_1 + |A_{0,1} + A_{0,2} + A_{1,2}|, -S_2 + |A_{0,1} + A_{0,2} - A_{1,2}|\}\\
    J_3 &\geq \frac{1}{8}\max\{-S_0 - S_1 + |B_0 - B_1|, -S_0 - S_2 + |B_0 - B_2|, -S_1 - S_2 + |B_1 - B_2|\}.
\end{align*}
Therefore at a given level of asynchronicity $S = \max\{S_0,S_1,S_2\}$, maximal violations of $J_3 \geq 0$ occur when $S_0 = S_1 = S_2$ and $A_{0,1} = A_{0,2} = A_{1,2} = 0$ and $B_0 = B_1 = B_2$. That is, for any given asynchronicity a maximal violation will occur within the symmetric nonsignaling polytope.

The symmetric nonsignaling polytope and its symmetric hidden variables subpolytope are much more tractable to analyze. In particular, the 48 facets indicated above reduce to the 12 facets 
\begin{equation}\label{eqn:symmetric-Bell-1}
J_j \geq -\frac{S_k}{4} \text{ (for $j=0,1,2,3$ and $k=0,1,2$)}
\end{equation}
of the symmetric hidden variables polytope. Another twelve related facets are $J_j \leq 1 + \frac{S_k}{4}$ ($j=0,1,2,3$ and $k=0,1,2$), but these play no role in our analysis. There are in total 72 additional facets to the symmetric hidden variables polytope that are not facets of the symmetric nonsignaling polytope (and hence are symmetric Bell inequalities). For completeness we present these in Appendix \ref{appendix:symnspoly}, however in addition to those in (\ref{eqn:symmetric-Bell-1}) there only twelve more that are useful here:
\begin{equation}\label{eqn:symmetric-Bell-2}
\begin{array}{ll}
|S_j - S_k| \leq 4(J_j + J_k) & \text{ (for $j,k\in\{0,1,2\}$ distinct)}\\
|S_j - S_k| \leq 4(J_\ell + J_3) & \text{ (for $j,k,\ell\in\{0,1,2\}$ distinct).}
\end{array}
\end{equation}

\begin{proposition}\label{proposition:symm-hidden-variables}
    Among symmetric hidden variables correlations, at most one of the inequalities $J_0,J_1,J_2,J_3 \geq 0$ can be violated. Moreover any such violation satisfies $J_j \geq \max\{-\frac{S_0}{4}, -\frac{S_1}{4}, -\frac{S_2}{4}\}$ and this bound is sharp.
\end{proposition}
\begin{proof}
    From (\ref{eqn:symmetric-Bell-2}) we have pairwise $J_j + J_k \geq 0$. Consequently if one $J_j < 0$, the other three must be positive. The given bound is just (\ref{eqn:symmetric-Bell-1}). One can easily construct a correlation that saturates this bound. For example, taking $U_0 = U_1 = U_2 = 0$, $S_0 = S_1 = S_2 = s$ and $C_{0,1} = C_{0,2} = C_{1,2} = -\frac{1}{3} - \frac{s}{3}$ has then $J_0 = J_1 = J_2 = \frac{1}{3} + \frac{s}{12}$ and $J_3 = -\frac{s}{4}$. All the inequalities (\ref{eqn:symmetric-Bell-1},\ref{eqn:symmetric-Bell-2}) and those in Appendix \ref{appendix:symnspoly} are satisfied.
\end{proof}

\begin{example}
    The correlation matrix indicated in Proposition \ref{proposition:symm-hidden-variables}  that saturates the bounds is
    
    \[\arraycolsep=3pt\def\arraystretch{1}
    \dfrac{1}{12}\left(\begin{array}{rrr|rrr|rrr}
    6-3s & 2 - s  & 2 - s  & 2 - s  & 6-3s & 2 - s  & 2 - s  & 2 - s  & 6-3s\\
    3s   & 4 + s  & 4 + s  & 4 + s  & 3s   & 4 + s  & 4 + s  & 4 + s  & 3s\\
    3s   & 4 + s  & 4 + s  & 4 + s  & 3s   & 4 + s  & 4 + s  & 4 + s  & 3s\\
    6-3s & 2 - s  & 2 - s  & 2 - s  & 6-3s & 2 - s  & 2 - s  & 2 - s  & 6-3s
    \end{array}\right)
    \]
    
    This matrix is realized by the classical hidden variables strategy given by the convex sum: $$\tfrac{s}{4} \left((f_0, f_7) + (f_7, f_0) \right) + \left(\tfrac{1}{6} - \tfrac{s}{12}\right)\left((f_{1}, f_{1}) + \cdots + (f_{6}, f_{6})\right)$$
    For $s = \frac{1}{2}$, this gives us the following correlation matrix which also has $J_3 = -\frac{1}{8}$:
    $$\frac{1}{8}\left(\begin{array}{ccc|ccc|ccc}
    3 & 1 & 1 & 1 & 3 & 1 & 1 & 1 & 3\\
    1 & 3 & 3 & 3 & 1 & 3 & 3 & 3 & 1\\
    1 & 3 & 3 & 3 & 1 & 3 & 3 & 3 & 1\\
    3 & 1 & 1 & 1 & 3 & 1 & 1 & 1 & 3
    \end{array}\right).$$
\end{example}

We indicate that this proposition holds for any hidden variables correlation. Namely, inequalities (\ref{eqn:symmetric-Bell-2}) are also facets of the full hidden variable correlation and analogous bounds to (\ref{eqn:symmetric-Bell-1}) can be derived (the bound for $J_3$ can be seen above).

\end{section}

\begin{section}{The causality loophole}\label{section:causality-loophole}

As we have proven above, the bounds $J_0, J_1, J_2, J_3 \geq -\frac{1}{8}$ are sharp and rigid among quantum synchronous correlations. However, this is not the case for more powerful strategies. For example, we have already seen these inequalities can be violated by general synchronous nonsignaling correlations. If (classical) communication between the two parties is allowed, then further violations can be achieved. This creates a ``causality loophole'' in our system: unless Alice and Bob are acausally separated, then the statistical tests above can be simply simulated using classical communication. Stated another way, if Eve has access to both Alice and Bob's inputs, then she can easily deliver output to each party that perfectly simulates the system above. In particular, whenever Alice and Bob have the same input values then Eve has full knowledge of the key bit she gives them.

\begin{example}
In order to simulate the correlation (\ref{eqn:J3correlation}) that achieves the maximal violation $J_3 = -\frac{1}{8}$ among synchronous quantum correlations, Eve does the following:
\begin{enumerate}
    \item she records Alice's input $x_A$ and delivers a uniformly random bit $y_A = y \in \{0,1\}$ to Alice;
    \item she checks if Bob's input $x_B$ is equal to $x_A$,
    \begin{itemize}
        \item if $x_B = x_A$ then she delivers the same output bit $y_B = y = y_A$ to Bob, or
        \item if $x_B \not= x_A$ then she picks $y_B = y$ with probability $\frac{1}{4}$, and $y_B = 1 - y$ with probability $\frac{3}{4}$, and delivers this output bit to Bob.
    \end{itemize}
\end{enumerate}
Note that Eve needs to communicate $(x_A,y_A)$ to ``Bob's side'' of the protocol to compute his output.
\end{example}

 Here we analyze a twist on the causality loophole: rather than limiting the communication Eve can perform, we assume she has imperfect knowledge of the Alice's and Bob's inputs. In the protocol, these inputs are random bases Alice and Bob independently select to perform a measurement. To be concrete, we will assume Eve's uncertainty is symmetric across all basis sections. In particular, if $\epsilon$ measures Eve's uncertainty, then
$$\mathrm{Pr}\{ \text{Eve guesses basis $x'$}\:|\: \text{Alice (or Bob) selects basis $x$}\} = \left\{\begin{array}{cl} 1 - \epsilon & \text{ when $x' = x$}\\\\ \frac{\epsilon}{2} & \text{ when $x' \not= x$.}\end{array}\right.$$
Eve, having made the guess $(z_A, z_B)$ at the bases used, has unlimited communication and computation and so produces outputs for Alice and Bob $(y_A,y_B)$ according to a correlation of her choosing, which we denote as $\Pr\{ (y_A,y_B) \:|\: (z_A,z_B)\}$. Therefore the correlation that Alice and Bob use for self-testing and deriving shared key is given by
$$p(y_A,y_B\:|\:x_A,x_B) = \sum_{z_A,z_B} \Pr\{ (y_A,y_B) \:|\: (z_A,z_B)\}\cdot
\left\{\begin{array}{cl} 1 - \epsilon & \text{for $z_A = x_A$} \\\\ \frac{\epsilon}{2} & \text{otherwise}\end{array}\right\}\cdot
\left\{\begin{array}{cl} 1 - \epsilon & \text{for $z_B = x_B$} \\\\ \frac{\epsilon}{2} & \text{otherwise}\end{array}\right\}.$$

Owing to the symmetry in Eve's uncertainty several of the observables that Alice and Bob compute have a very simple form. However since Eve is not restricted to nonsignaling strategies, we must use forms suitable for general correlations. See Appendix \ref{appendix:notes} for these. In particular, the the expected asymmetry is given by
\begin{align*}
    \langle A_{0,1} \rangle &= p(0,1\:|\:1,0) - p(1,0\:|\:0,1) + p(1,0\:|\:1,0) - p(0,1\:|\:0,1)\\
    &= \left(1 - 2\epsilon + \tfrac{3}{4}\epsilon^2\right) \tilde{A}_{0,1} +  \left(\tfrac{1}{2}\epsilon - \tfrac{3}{4}\epsilon^2\right) (\tilde{A}_{0,2} - \tilde{A}_{1,2})
\end{align*}
where $\tilde{A}$ are the analogous assymmetry terms for Eve's strategy $\Pr$. Note that when Eve has no knowledge of Alice and Bob's basis selection, $\epsilon = \frac{2}{3}$, then Alice and Bob's received output from Eve is independent of their basis selection and so no asymmetry is observed. However as $\epsilon$ decreases any asymmetries in Eve's strategy start to become evident to Alice and Bob. A similar situation occurs for the biases; for example
\begin{align*}
    \langle B_{0} - B_{1} \rangle &= 2(p(0,0\:|\:0,1) - p(1,1\:|\:0,1) - p(0,0\:|\:1,0) + p(1,1\:|\:1,0))\\
    &= \left(1 - 2\epsilon + \tfrac{3}{4}\epsilon^2\right) \left(\tilde{B}_{0} - \tilde{B}_{1}\right) +  \left(\tfrac{1}{2}\epsilon - \tfrac{3}{4}\epsilon^2\right) \left(\left(\tilde{B}_{0} - \tilde{B}_{2}\right) - \left(\tilde{B}_{1} - \tilde{B}_{2}\right)\right)\\
    &= \left(1 - \tfrac{3}{2}\epsilon\right)\left(\tilde{B}_{0} - \tilde{B}_{1}\right)
\end{align*}

As for the asynchronicity measure $S_0 + S_1 + S_2$, we compute the expected value as follows:
\begin{align}
    \langle S_0 + S_1 + S_2 \rangle &= 2 \left(p(0,1\:|\:0,0) + p(1,0\:|\:0,0) + p(0,1\:|\:1,1) + p(1,0\:|\:1,1) + p(0,1\:|\:2,2) + p(1,0\:|\:2,2)\right) \notag \\
    &= \left(1 - 2\epsilon + \tfrac{3}{2}\epsilon^2\right)\left(\tilde{S}_0 + \tilde{S}_1 + \tilde{S}_2\right) + \left(8\epsilon - 6\epsilon^2\right)(1-\tilde{J}_3) \label{eqn:expectedS}
\end{align}
where $\tilde{S}$ and $\tilde{J}$ are the analogous asynchronicity and $J$ terms, respectively, for Eve's strategy.

\noindent Similarly for $J_3$, we find that the expected value is
\begin{align}
    \langle 1 - J_3 \rangle &= 1 - \frac{1}{4}\left(1 + C_{0,1} + C_{0, 2} + C_{1,2}\right) \notag\\
        &=1 - \tfrac{1}{4}\left( p(0,1\:|\:0,1) + p(1,0\:|\:0,1) + p(0,1\:|\:1,0) + p(1,0\:|\:1,0)\right. \notag\\
        &\qquad\qquad +\ p(0,1\:|\:0,2) + p(1,0\:|\:0,2) + p(0,1\:|\:2,0) + p(1,0\:|\:2,0) \notag \\
        &\qquad\qquad +\ \left.p(0,1\:|\:1,2) + p(1,0\:|\:1,2) + p(0,1\:|\:2,1) + p(1,0\:|\:2,1)\right) \notag\\
    &=  \left(1 - \epsilon + \tfrac{3}{4}\epsilon^2\right)(1-\tilde{J}_3) + \left(\tfrac{1}{4}\epsilon - \tfrac{3}{16}\epsilon^2\right)\left(\tilde{S}_0 + \tilde{S}_1 + \tilde{S_2}\right) \label{eqn:expectedJ0}
\end{align}

\noindent Let $0 \leq \lambda \leq \frac{1}{8}$ and $0 \leq \mu \leq \mu_0$ be the allowed errors in the expected values of $J_3$ and $S_0 + S_1 + S_2$ respectively. Using equations \ref{eqn:expectedS} and \ref{eqn:expectedJ0}, we can compute the values of $\tilde{S} = \tilde{S}_0 + \tilde{S}_1 + \tilde{S}_2$ and $\tilde{J}_3$.

\begin{align*}
    \begin{bmatrix}
        1 - \tilde{J}_3 \\
        \tilde{S}
    \end{bmatrix}
    &=
    \begin{bmatrix*}[r]
        1 - \epsilon + \tfrac{3}{4}\epsilon^2 & \tfrac{1}{4}\epsilon - \tfrac{3}{16}\epsilon^2 \vspace{0.2cm} \\
        8\epsilon - 6\epsilon^2 & 1 - 2\epsilon + \tfrac{3}{2}\epsilon^2
    \end{bmatrix*}^{-1}
    \begin{bmatrix}
        \frac{9}{8} - \lambda \\
        \mu
    \end{bmatrix}
\end{align*}

This gives us solutions,
\begin{align}
    \label{eqn:eve's-J3} \tilde{J}_3 &= 1 - \frac{ (3\epsilon^2 - 4\epsilon)(\mu  - 8\lambda + 9) - 16\lambda + 18}{4\left(3\epsilon - 2\right)^{2}} = \frac{ (3\epsilon^2 - 4\epsilon)(3 - \mu  + 8\lambda) - 16\lambda + 34}{4\left(3\epsilon - 2\right)^{2}}\\
    \label{eqn:eve's-S} \tilde{S} &= \frac{(3\epsilon^2 - 4\epsilon)(\mu - 8\lambda + 9) + 4\mu}{\left(3\epsilon - 2\right)^{2}}.
\end{align}

\noindent We claim that for a maximum value, say $\mu_0$, for Eve's asynchronicity, her uncertainty cannot grow too much before the asynchronicity becomes negative, which will result in an infeasible strategy. We solve for $\epsilon_{max}$ when $\tilde{S} = 0.01$ which gives us $$\epsilon_{max} = \frac{2}{3} - \frac{2}{3}\left(\frac{10\sqrt{6400\lambda^2 + 2(400\lambda - 447)\mu - 200\mu^2 - 14376\lambda + 8073}}{100\mu - 800\lambda + 897}\right)$$


We plot the value of $\epsilon$ against varying values of $\mu_0$ in Figure \ref{fig:asynchronicity_v_epsilon}.  


In a different direction, the analysis above also allows us show robustness of the protocol under the presence of a constant amount of noise, we use two kinds of noise models: (i) a depolarizing channel, and (ii) slew errors in the device measurement angles.

A depolarizing channel acting on state $\varrho$ is defined as
$$\upxi (\varrho) = \varrho' = (1-\eta)\varrho + \frac{\eta}{3}\left(X\varrho X + Y\varrho Y + Z\varrho Z\right).$$
The channel leaves the operator $\rho$ fixed with probability $1-\eta$, and applies one of the Pauli gates $X, Y$ or $Z$ with probability $\frac{\eta}{3}$.
In the protocol, we use an EPR pair, $\ket{\psi} = \frac{1}{\sqrt{2}}\left(\ket{00} + \ket{11}\right)$. Therefore, we have $$\rho = \ket{\psi}\bra{\psi} = \frac{1}{2}\left(\ket{00}\bra{00} + \ket{00}\bra{11} + \ket{11}\bra{00} + \ket{11}\bra{11}\right) = \begin{bmatrix}\frac{1}{2} & 0 & 0 & \frac{1}{2} \\ 0 &0 &0 &0 \\ 0 &0 &0 &0 \\ \frac{1}{2} &0 &0 &\frac{1}{2} \end{bmatrix}$$
Applying the channel independently to Alice's and Bob's qubits, we get,
\begin{align*}
        \rho' = (1-\eta)^2\rho 
         &+ (1-\eta)\frac{\eta}{3} \left( \sum_{\sigma \in \{X,Y,Z\}}{(\mathbb{1}\otimes \sigma)\rho(\mathbb{1}\otimes \sigma)^\dag + (\sigma \otimes \mathbb{1})\rho(\sigma \otimes \mathbb{1})^\dag} \right) \\
         &+\frac{\eta^2}{9} \left( \sum_{\sigma, \sigma' \in \{X, Y, Z\}}{(\sigma \otimes \sigma')\rho(\sigma \otimes \sigma')^\dag} \right)
\end{align*}
Therefore, our probability distribution after applying the channel is given by
$p'(y_A, y_B|x_A, x_B) =\tr(\rho' (E^{x_A}_{y_A} \otimes E^{x_B}_{y_B}))$, where $E^x_y$ are the projections from observables (\ref{eqn:J3measurement}). The expected value of $J_3$ given as follows:\\
\begin{align*}
    \langle J_3' \rangle &= 1 - \frac{1}{4}(p'(0,1|0,1) + p'(0,1|1,0) + p'(1,0|0,1) + p'(1,0|1,0) + p'(0,2|0,1) + p'(0,2|1,0)  \\
    & { } + p'(2,0|0,1) + p'(2,0|1,0) + p'(1,2|0,1) + p'(1,2|1,0) + p'(2,1|0,1) + p'(2,1|1,0))\\
    \therefore \langle J_3' \rangle &= - \frac{1}{8} + \eta -\frac{2}{3}\eta^2
\end{align*}

Similarly, we compute the asynchronicity as $\langle S' \rangle = \langle S_0' + S_1' + S_2' \rangle = 8\eta - \frac{16}{3}\eta^2$. We insert these expected error values into equations (\ref{eqn:eve's-J3},\ref{eqn:eve's-S}) and obtain
\begin{align*}
    \tilde{J}_3 &= 1 - \frac{9(3\epsilon^2 - 4\epsilon) - 16\eta + \frac{32}{3}\eta^2 + 18}{4\left(3\epsilon - 2\right)^{2}}  \\
    \tilde{S} &= \frac{9(3\epsilon^2 - 4\epsilon) + 32\eta - \frac{64}{3}\eta^2}{\left(3\epsilon - 2\right)^{2}}.
\end{align*}
This gives
$$\epsilon_{max} = \frac{2}{3} -\frac{20\sqrt{897}}{2961} \left(\sqrt{16\eta^2 - 24\eta + 9}\right)$$

 Next, we turn to slew errors in measurement angles, and study their effect on the observed value of $J_3$, which Alice and Bob use to detect interference by Eve. Even if the devices that Alice and Bob use are programmed according to their specification, it is unlikely that the devices measure exactly along the angles specified. It is difficult to produce perfect devices, and thus it is likely that we have errors in the measurement angles. In order to account for this type of noise, we simulate the protocol in Microsoft Q\# \cite{qsharp}, and add to the measurement angles Gaussian noise with mean zero and standard deviation 1/8. Table \ref{table:eta-and-J3} summarizes the observed value of $J_3$ over varying values of $\eta$, which is the probability of error for the depolarization channel described above. We observe that the protocol can tolerate even as much as 5\% probability of error($\eta$) in the depolarization channel along with slew errors, while still producing a large $J_3$ violation.
 \newline

\begin{minipage}[b]{.47\textwidth}
    \centering
    \begin{tikzpicture}[y=.7cm, x=1cm]
        \draw (0,0) -- coordinate (x axis mid) (6,0) node[anchor=north west] {$\mu_0$};
        \draw (0,0) -- coordinate (y axis mid) (0,7) node[anchor=south east] {$\epsilon_{max}$};
        \draw (0, 1pt) -- (0, -3pt) node[anchor=north] {$0$};
        \foreach \x in {1, 2, 3, 4, 5}
            \draw (\x, 1pt) -- (\x, -3pt) node[anchor=north] {$0.0\x$};
        \foreach \y in {1, 2, 3, 4, 5, 6}
            \draw (1pt, \y) -- (-3pt, \y) node[anchor=east] {$\y\times 10^{-3}$};
    	\draw plot[mark=*, mark options={fill=magenta}]
    	    file {eps-mu1.data};
    	\begin{scope}[shift={(4,4)}] 
            \draw (-3.5,3.0) node[right]{$\tilde{S} = 1\%, \lambda = \frac{1}{8}$};
	    \end{scope}
    \end{tikzpicture}
    \captionof{figure}{Values of $\mu_0$ vs. $\epsilon_{max}$ for which Eve's asynchronicity $\tilde{S}$ is positive}
    \label{fig:asynchronicity_v_epsilon}
\end{minipage} \qquad
\begin{minipage}[b]{.47\textwidth}
    \centering
    \begin{tabular}{|c|c|c|}
    \hline
      \textbf{ $\eta$ }& \textbf{$J_3$} & \textbf{$J_3$ variance}\\ \hline
      $0.030$ & $-0.097$ & $5.9 \times 10^{-5}$\\ \hline
      $0.032$ & $-0.088$ & $5.92 \times 10^{-5}$\\ \hline
      $0.034$ & $-0.084$ & $5.94 \times 10^{-5}$\\ \hline
      $0.036$ & $-0.078$ & $5.96 \times 10^{-5}$\\ \hline
      $0.038$ & $-0.076$ & $5.97 \times 10^{-5}$\\ \hline
      $0.040$ & $-0.089$ & $5.99 \times 10^{-5}$\\ \hline
      $0.042$ & $-0.076$ & $6.01\times 10^{-5}$\\ \hline
      $0.044$ & $-0.082$ & $6.02 \times 10^{-5}$\\ \hline
      $0.046$ & $-0.083$ & $6.04 \times 10^{-5}$\\ \hline
      $0.048$ & $-0.068$ & $6.06 \times 10^{-5}$\\ \hline
      $0.050$ & $-0.069$ & $6.07 \times 10^{-5}$\\ \hline
    \end{tabular}
    \captionof{table}{Observed $J_3$ statistics over varying values of $\eta$}
    \label{table:eta-and-J3}
\end{minipage}

\end{section}

\appendix

\newpage
\begin{section}{Classical and Quantum Synchronous Correlations}\label{appendix:extreme-points}

Recall a \emph{local hidden variables strategy}, or simply \emph{classical correlation}, is a correlation of the form
    \begin{equation}\label{eqn:classical:definition-appendix}
        p(y_A,y_B\:|\:x_A,x_B) = \sum_{\omega\in\Omega} \mu(\omega) p_A(y_A\:|\: x_A,\omega) p_B(y_B\:|\: x_B,\omega)
    \end{equation}
for some finite set $\Omega$ and probability distribution $\mu$. Here $(\Omega,\mu)$ is shared randomness Alice and Bob may draw on, and $p_A$ and $p_B$ are local (conditional) probabilities they use to produce their respective outputs. Clearly every correlation of the form (\ref{eqn:classical:definition-appendix}) will be nonsignaling. Without loss of generality, we may assume $\mu(\omega) > 0$ for all $\omega\in \Omega$ as otherwise we simply restrict to the support of $\mu$ and still have the same form.

In order that $p$ be synchronous we must have for each $x\in X$, whenever $y_A \not= y_B$ that
$$0 = \sum_{\omega\in\Omega} \mu(\omega) p_A(y_A\:|\: x,\omega) p_B(y_B\:|\: x,\omega).$$
In particular, for each $x\in X$ and $\omega\in \Omega$ we must have whenever $y_A \not= y_B$ that
$$0 = p_A(y_A\:|\: x,\omega) p_B(y_B\:|\: x,\omega).$$
Fix some $x,\omega$. As $\sum_y p_A(y\:|\:x,\omega) = 1$ there exists a $y_0$ (depending on $x,\omega$) such that $p_A(y_0\:|\: x,\omega) > 0$, and so from above $p_B(y\:|\:x,\omega) = 0$ whenever $y \not= y_0$. Thus
$$1 = \sum_y p_B(y\:|\:x,\omega) = p_B(y_0\:|\: x,\omega).$$
Exchanging $A$ and $B$ shows $p_A(y_0\:|\: x,\omega) = 1$ as well.

Therefore, for each $\omega\in\Omega$ we obtain a function $f_\omega:X \to Y$ given by $f_\omega(x) = y_0$ where $y_0$ is the value with $p_A(y_0\:|\: x,\omega) = p_B(y_0\:|\: x,\omega) = 1$. This allows us to map $\Omega$ into the set of function $X\to Y$, proving the following results.

\setcounter{theorem}{2}
\begin{theorem}
The set of synchronous classical correlations with input $X$ and output $Y$ is bijective to the collection of probability distributions on the set of functions $X \to Y$. Given such a probability distribution, the associated strategy is: Alice and Bob sample a function $f:X\to Y$ according the specified distribution, and upon receiving $x_A,x_B$ they output $y_A = f(x_A)$ and $y_B = f(x_B)$.
\end{theorem}

\begin{corollary}
    The extreme points of the synchronous hidden variables strategies from $X$ to $Y$ can be canonically identified with the set of functions $X \to Y$.
\end{corollary}

\begin{corollary}
    Every synchronous classical strategy is symmetric.
\end{corollary}

A \emph{quantum correlation} is a correlation that takes the form
    \begin{equation}\label{eqn:quantum:definition}
        p(y_A,y_B\:|\:x_A,x_B) = \tr(\rho(E^{x_A}_{y_A}\otimes F^{x_B}_{y_B}))
    \end{equation}
where $\rho$ is a density operator on the Hilbert space $\h{H}_A\otimes\h{H}_B$, and for each $x\in X$ we have $\{E^x_y\}_{y\in Y}$ and $\{F^x_y\}_{y\in Y}$ are POVMs on $\h{H}_A$ and $\h{H}_B$ respectively.  Again, any correlation of the form (\ref{eqn:quantum:definition}) will be nonsignaling from the fact that $\{E^x_y\}_{y\in Y}$ and $\{F^x_y\}_{y\in Y}$ are POVMs. We will only treat the case when $\h{H}_A$ and $\h{H}_B$ are finite dimensional. Without loss of generality we can take $\tr_A(\rho)$ and $\tr_B(\rho)$ of maximal rank by restricting $\h{H}_A$ and $\h{H}_B$ if necessary. It is common to define quantum correlations using projections rather than general positive-operator valued measures as one can always enlarge $\h{H}_A$ and $\h{H}_B$ to achieve such. However for synchronous quantum correlations we have these POVMs must already be projection-valued. The proof of this is contained in \cite[Proposition 1]{cameron2007quantum}, which is for ``quantum coloring games'' but carries over to synchronous quantum correlations without modification; see also \cite{abramsky2017quantum,atserias2016quantum,manvcinska2016quantum,paulsen2016estimating}.

\begin{lemma}
    Let $p(y_A,y_B\:|\:x_A,x_B) = \tr(\rho(E^{x_A}_{y_A}\otimes F^{x_B}_{y_B}))$ be a synchronous quantum correlation. Then the POVMs $\{E^x_y\}_{y\in Y}$ and $\{F^x_y\}_{y\in Y}$, for $x\in X$, are projection-valued measures. Moreover each $E^x_y$ commutes with $\tr_B(\rho)$ and each $F^x_y$ commutes with $\tr_A(\rho)$.
\end{lemma}

The works cited above are primarily focused on the existence of a synchronous quantum correlation that satisfies some additional conditions, for example preserving graph adjacency. In this context, a common result is that if one such correlation exists then another exists whose state is maximally entangled; examples of such include \cite[Proposition 1]{cameron2007quantum}, \cite[Lemma 4]{abramsky2017quantum}, \cite[Theorem 2.1]{manvcinska2016quantum}. It is certainly not the case that every synchronous quantum correlation can be taken to have a maximally entangled state, as these include hidden variables strategies. Nonetheless we can prove that every synchronous quantum correlation is a convex sum of such. 

\begin{lemma}
    Every synchronous quantum correlation can be expressed as the convex combination of synchronous quantum correlations with maximally entangled pure states. In particular, if a synchronous quantum correlation $\tr(\rho(E^{x_A}_{y_A}\otimes F^{x_B}_{y_B}))$ is extremal then we may take $\rho = \ket\psi\bra\psi$ with $\ket\psi$ maximally entangled.
\end{lemma}
\begin{proof}
    Let $p(y_A,y_B\:|\:x_A,x_B) = \tr(\rho(E^{x_A}_{y_A}\otimes F^{x_B}_{y_B}))$ be a synchronous quantum correlation. As we may decompose $\rho$ into a convex combination of pure states, we can assume $\rho = \ket\psi\bra\psi$. Suppose $\ket\psi$ has $r$ distinct Schmidt coefficients, and in particular let us write the Schmidt decomposition of $\ket\psi$ as
    $$\ket\psi = \sum_{j=1}^r \sqrt{\sigma_j}\sum_{m=1}^{\ell_j} \ket{\phi_{j,m}^A}\otimes\ket{\phi_{j,m}^B}.$$
    Note $\sum_j\ell_j\sigma_j = 1$. The spectral decomposition of the partial trace is then
    $$\tr_B(\ket\psi\bra\psi) = \sum_{j=1}^r \sigma_j \Pi_j^A,$$
    where
    $$\Pi_j^A = \sum_{m=1}^{\ell_j} \ket{\phi_{j,m}^A}\bra{\phi_{j,m}^A}.$$
    These eigenprojections $\{\Pi_j^A\}$ decompose Alice's Hilbert space $\h{H}_A$ into an orthogonal sum of subspaces: $\h{H}_A = \h{H}_0 \oplus \bigoplus \h{H}_j^A$ where $\h{H}_j^A = \mathrm{im}(\Pi_j^A)$ and $\h{H}_0^A = \ker(\tr_B(\ket\psi\bra\psi))$. Identically $\tr_A(\ket\psi\bra\psi) = \sum_{j=1}^r \sigma_j \Pi_j^B$, inducing a decomposition of Bob's space $\h{H}_B = \h{H}_0 \oplus \bigoplus \h{H}_j^B$.
    
    From the lemma, each $E^x_y$ commutes with $\tr_A(\ket\psi\bra\psi)$ and so preserves the decomposition $\h{H}_A = \h{H}_0 \oplus \bigoplus \h{H}_j^A$, with a similar statement holding for the $F^x_y$. Therefore
    \begin{align*}
        &p(y_A,y_B\:|\:x_A,x_B) = \bra\psi E^{x_A}_{y_A}\otimes F^{x_B}_{y_B}\ket\psi\\
        &= \sum_{j,k=1}^r \sqrt{\sigma_j\sigma_k} \sum_{m,n=1}^{\ell_j,\ell_k} \bra{\phi_{j,m}^A} E^{x_A}_{y_A} \ket{\phi_{k,n}^A}\bra{\phi_{j,m}^B} F^{x_B}_{y_B} \ket{\phi_{k,n}^B}\\
        &= \sum_{j=1}^r \sigma_j \sum_{m,n=1}^{\ell_j,\ell_k} \bra{\phi_{j,m}^A} E^{x_A}_{y_A} \ket{\phi_{j,n}^A}\bra{\phi_{j,m}^B} F^{x_B}_{y_B} \ket{\phi_{j,n}^B}\\
        &= \sum_{j=1}^r \ell_j\sigma_j \bra{\psi_j} E^{x_A}_{y_A}\otimes F^{x_B}_{y_B} \ket{\psi_j},
    \end{align*}
    where $\ket{\psi_j} = \frac{1}{\sqrt{\ell_j}} \sum_{m=1}^{\ell_j} \ket{\phi_{j,m}^A}\otimes\ket{\phi_{j,m}^B}$ is maximally entangled on $\h{H}_j^A\otimes\h{H}_j^B$.
\end{proof}

This theorem shows that any extremal synchronous quantum correlation will be associated to some maximally entangled pure state, which can be taken canonically \cite{werner2001all}. After restricting $\h{H}_A$ and $\h{H}_B$ to the support of the partial traces of $\ket\psi\bra\psi$ if necessary, we can take
$$\ket\psi = (V\times\openone)\ket\Omega \text{ where } \ket\Omega = \frac{1}{\sqrt{d}} \sum_{j=1}^d \ket{j,j}.$$
Here $\{\ket{j}\}_{j=1}^d$ is a fixed orthonormal basis of $\h{H} = \h{H}_B$ and $V$ is an isometry of $\h{H}_A$ onto $\h{H}_B$, or unitary upon also identifying $\h{H}_A = \h{H}$. Direct manipulation of the resulting expression leads to the following result, which can be found in greater generality as \cite[Theorem 5.5]{paulsen2016estimating}. 

\begin{theorem}
Let $X,Y$ be finite sets, $\h{H}$ a $d$-dimensional Hilbert space, and for each $x\in X$ a projection-valued measure $\{E^x_y\}_{y\in Y}$ on $\h{H}$. Then
$$p(y_A,y_B\:|\: x_A,x_B) = \frac{1}{d}\tr(E^{x_A}_{y_A}E^{x_B}_{y_B})$$
defines a synchronous quantum correlation. Moreover every synchronous quantum correlation with maximally entangled pure state has this form.
\end{theorem}
\begin{proof}
Suppose $p(y_A,y_B\:|\: x_A,x_B) = \bra{\psi}(E^{x_A}_{y_A}\otimes F^{x_B}_{y_B})\ket{\psi}$ with $\ket\psi$ maximally entangled. Then as noted before  $\ket\psi = (V\times\openone)\ket\Omega$ where $\ket\Omega = \frac{1}{\sqrt{d}} \sum_{j=1}^d \ket{j,j}$ with  $\{\ket{j}\}_{j=1}^d$ a fixed orthonormal basis of $\h{H} = \h{H}_B$ and $V$ is an isometry of $\h{H}_A$ onto $\h{H}_B$ (or unitary if we identify $\h{H}_A = \h{H}$). So we have
$$p(y_A,y_B\:|\: x_A,x_B) = \bra{\Omega}(V^\dagger E^{x_A}_{y_A}V\otimes F^{x_B}_{y_B})\ket{\Omega}$$

Redefining $V^\dagger E^x_y V \mapsto E^x_y$ reduces our form for correlations with maximally entangled state to
\begin{align*}
    p(y_A,y_B\:|\: x_A,x_B) &= \bra\Omega E^{x_A}_{y_A}\otimes F^{x_B}_{y_B} \ket\Omega\\
    &= \frac{1}{d} \sum_{j,k=1}^d \bra{j} E^{x_A}_{y_A} \ket{k}\bra{j} F^{x_B}_{y_B} \ket{k}\\
    &= \frac{1}{d} \sum_{j,k=1}^d \bra{j} E^{x_A}_{y_A} \ket{k}\bra{k} \overline{F^{x_B}_{y_B}} \ket{j}\\
    &= \frac{1}{d} \tr(E^{x_A}_{y_A} \overline{F^{x_B}_{y_B}}).
\end{align*}
Here $\overline{F^{x}_{y}}$ refers to the projection whose entries in the $\{\ket{j}\}_{j=1}^d$ basis are the complex conjugates of those of $F^{x}_{y}$; that is $\overline{F^{x}_{y}}$ is the transpose of $F^{x}_{y}$ with respect to this basis.

As $p$ is synchronous we have
$$1 = \frac{1}{d} \sum_{y_A,y_B} \tr(E^{x}_{y_A} \overline{F^{x}_{y_B}}) = \frac{1}{d} \sum_{y} \tr(E^{x}_{y} \overline{F^{x}_{y}}).$$ 
But from Cauchy-Schwarz,
\begin{align*}
    1 &= \frac{1}{d} \sum_{y} \tr(E^{x}_{y} \overline{F^{x}_{y}})\\
    &\leq \left[\frac{1}{d} \sum_{y} \tr(E^{x}_{y} E^{x}_{y})\right]^{\frac{1}{2}}\left[\frac{1}{d} \sum_{y} \tr(\overline{F^{x}_{y}} \overline{F^{x}_{y}})\right]^{\frac{1}{2}}\\
    &= \left[\tfrac{1}{d} \tr(\openone)\right]^{\frac{1}{2}}\left[\tfrac{1}{d} \tr(\openone)\right]^{\frac{1}{2}} = 1.
\end{align*}
And so again by Cauchy-Schwarz $E^{x}_{y} = \overline{F^{x}_{y}}$. 

As for the converse, synchronicity is clear as $\frac{1}{d}\tr(E^x_{y_A}E^x_{y_B}) = 0$ if $y_A \not= y_B$ as $\{E^x_{y}\}_{y\in Y}$ is a projection-valued measure. To see it is quantum, we reverse the computation above and write
$$\frac{1}{d}\tr(E^x_{y_A}E^x_{y_B}) = \bra\Omega E^{x_A}_{y_A}\otimes \overline{E}^{x_B}_{y_B} \ket\Omega.$$
\end{proof}

\begin{corollary}
    Every synchronous quantum correlation is symmetric.
\end{corollary}

\end{section}

\begin{section}{Synchronous correlations with \texorpdfstring{$|X| = 2$}{Lg}}

Let us consider with the case of correlations with $|X|=2$ into some finite set $Y$. For concreteness we take $X = \{0,1\}$. In the case of general (not necessarily synchronous) nonsignaling correlations, some results along these lines are known \cite{barrett2005nonlocal}. However synchronous correlations have a very rigid structure, which allows us to characterize them completely.  In the first lemma, we show that a general synchronous nonsignalling correlation is characterized by two interrelated functions $u(y_A,y_B),\ v(y_A,y_B)$ on $Y^2$. The second lemma shows that the local hidden variables correlations have the same structure with $u(y_A,y_B) = v(y_B,y_A)$. Therefore, in the context of synchronous nonsignalling correlations, symmetric and classical are equivalent. As all synchronous quantum correlations are symmetric, there can be no quantum correlation that are not classical, and hence no synchronous Bell inequalities when $X = \{0,1\}$. 

\setcounter{theorem}{17}
\begin{lemma}\label{lemma:2-point-domain:nonsignaling}
    Let $Y$ be a finite set and $u = u(y_A,y_B)$ and $v = v(y_A,y_B)$ be probability distributions on $Y^2$ such that for all $y \in Y$
    $$\sum_{y'} u(y,y') = \sum_{y'} v(y',y) \text{ and }\sum_{y'} u(y',y) = \sum_{y'} v(y,y').$$
    Write $\theta(y)$ and $\phi(y)$ for these two sums respectively and define
    $$\begin{array}{rl@{\qquad}rl}
        p(y_A,y_B\:|\:0,0) &= \indicator{y_A=y_B}\theta(y_A), & p(y_A,y_B\:|\:0,1) &= u(y_A,y_B),\\
        p(y_A,y_B\:|\:1,1) &= \indicator{y_A=y_B}\phi(y_A), & p(y_A,y_B\:|\:1,0) &= v(y_A,y_B).\\
    \end{array}$$
    Then $p$ is a synchronous nonsignaling correlation. Moreover every nonsignaling correlation with domain $\{0,1\}$ arises this way.
\end{lemma}
\begin{proof}
    A straightforward computation shows that $p$ as defined is synchronous and satisfies the nonsignaling conditions. Conversely, given a synchronous nonsignaling correlation $p$ from $\{0,1\}$ to $Y$, define 
    $$u(y_A,y_B) = p(y_A,y_B\:|\:0,1) \text{ and } v(y_A,y_B) = p(y_A,y_B\:|\:1,0).$$
    Then
    $$\sum_{y'} p(y,y'\:|\:0,0) = \sum_{y'} p(y,y'\:|\:0,1) = \sum_{y'} u(y,y') = \theta(y)$$
    where we take this as the definition of $\theta$. Then $p(y,y'\:|\:0,0) = 0$ when $y\not= y'$ and
    $$p(y,y\:|\:0,0) = \sum_{y'}p(y,y'\:|\:0,0) = \theta(y)$$
    and therefore $p(y,y'\:|\:0,0) = \indicator{y=y'}\theta(y)$. An identical argument shows $p(y,y'\:|\:1,1) = \indicator{y=y'}\phi(y)$ where $\phi(y) = \sum_{y'}v(y,y')$. Finally, we compute
    \begin{align*}
        \sum_{y'} u(y,y') &= \sum_{y'} p(y,y'\:|\:0,1) = \sum_{y'} p(y,y'\:|\:0,0) = p(y,y\:|\: 0,0)\\ 
        &= \sum_{y'} p(y',y\:|\:0,0) = \sum_{y'} p(y',y\:|\:1,0) = \sum_{y'} v(y',y).
    \end{align*}
    Again a similar argument shows $\sum_{y'} u(y',y) = \sum_{y'} v(y,y')$, and therefore $u$ and $v$ satisfy the two constraints stated in the theorem.
\end{proof}

\begin{lemma}\label{lemma:2-point-domain:classical}
    Let $Y$ be a finite set and $u = u(y_A,y_B)$ be probability distributions on $Y^2$. Write
        $$\theta(y) = \sum_{y'} u(y,y') \text{ and } \phi(y) = \sum_{y'} u(y',y).$$
    Define
    $$\begin{array}{rl@{\qquad}rl}
        p(y_A,y_B\:|\:0,0) &= \indicator{y_A=y_B}\theta(y_A), & p(y_A,y_B\:|\:0,1) &= u(y_A,y_B),\\
        p(y_A,y_B\:|\:1,1) &= \indicator{y_A=y_B}\phi(y_A), & p(y_A,y_B\:|\:1,0) &= u(y_A,y_B).\\
    \end{array}$$
    Then $p$ is a synchronous classical correlation. Moreover every nonsignaling correlation with domain $\{0,1\}$ arises this way.
\end{lemma}
\begin{proof}
    Define a probability distribution on the functions $\{0,1\} \to Y$ by $\mu(f) = u(f(0),f(1))$. Then
    $$\sum_f u(f(0),f(1)) \indicator{y_A = f(0)}\indicator{y_B = f(1)} = u(y_A,y_B) = p(y_A,y_B\:|\:0,1)$$
    and identically 
    $$p(y_A,y_B\:|\:1,0) = \sum_f u(f(0),f(1)) \indicator{y_A = f(1)}\indicator{y_B = f(0)}.$$
    Similar computation show
    $$\sum_f u(f(0),f(1)) \indicator{y_A = f(0)}\indicator{y_B = f(0)} = \indicator{y_A=y_B}\sum_{y'} u(y_A,y') =  p(y_A,y_B\:|\:0,0)$$
    and
    $$p(y_A,y_B\:|\:1,1) = \sum_f u(f(0),f(1)) \indicator{y_A = f(1)}\indicator{y_B = f(1)}.$$    Therefore $p$ is the classical strategy associated to the distribution $\mu$.
    
    Conversely, if $p$ is classical then it is nonsignaling and so has the form as given in Lemma \ref{lemma:2-point-domain:nonsignaling}. But from Corollary \ref{corollary:classical:symmetric}, $p$ is also symmetric and hence $v(y_A,y_B) = u(y_B,y_A)$.
\end{proof}

\begin{corollary}
    A synchronous correlation with $|X| = 2$ is classical if and only if it is symmetric.
\end{corollary}

\begin{theorem}
    Every synchronous quantum correlation with $|X| = 2$ to any finite set is classical.
\end{theorem}

These results imply that for synchronous correlations with $|X| = 2$ the Bell inequalities are actually equations: the equations for a correlation being symmetric. Consequently, one cannot achieve a Bell violation in this circumstance since all synchronous quantum correlations will also satisfy these equations.

\end{section}

\begin{section}{Example: \texorpdfstring{$|X| = |Y| = 2$}{Lg}} \label{appendix:case2-2}

In this appendix we explicitly express several polytopes of nonlocal games using the biases and correlation matrix above, and focus on the case of $|X| = |Y| = 2$. Starting with the nonsignaling polytope, we have seen in (\ref{eqn:nonsignaling-inequalities}) above that this can be described as the inequalities 
\begin{equation}
    \begin{array}{rl}
    1 + a_{j} + b_{k} + c_{j,k} &\geq 0,\\
    1 - a_{j} - b_{k} + c_{j,k} &\geq 0,\\
    1 - a_{j} + b_{k} - c_{j,k} &\geq 0,\\
    1 + a_{j} - b_{k} - c_{j,k} &\geq 0,
    \end{array}
\end{equation}
with indices ranging $j,k = 0,1$. This defines an 8-dimensional polyhedron in $\mathbb{R}^8$. For completeness we write these out explicitly:
\begin{align*}
1 + a_{0} + b_{0} + c_{0,0} \geq 0&,\quad 1 - a_{0} - b_{0} + c_{0,0} \geq 0,\\
1 - a_{0} + b_{0} - c_{0,0} \geq 0&,\quad 1 + a_{0} - b_{0} - c_{0,0} \geq 0,\\
1 + a_{0} + b_{1} + c_{0,1} \geq 0&,\quad 1 - a_{0} - b_{1} + c_{0,1} \geq 0,\\
1 - a_{0} + b_{1} - c_{0,1} \geq 0&,\quad 1 + a_{0} - b_{1} - c_{0,1} \geq 0,\\
1 + a_{1} + b_{0} + c_{1,0} \geq 0&,\quad 1 - a_{1} - b_{0} + c_{1,0} \geq 0,\\
1 - a_{1} + b_{0} - c_{1,0} \geq 0&,\quad 1 + a_{1} - b_{0} - c_{1,0} \geq 0,\\
1 + a_{1} + b_{1} + c_{1,1} \geq 0&,\quad 1 - a_{1} - b_{1} + c_{1,1} \geq 0,\\
1 - a_{1} + b_{1} - c_{1,1} \geq 0&,\quad 1 + a_{1} - b_{1} - c_{1,1} \geq 0.
\end{align*}

The hidden variables polytope is an 8-dimensional subpolyhedron of the nonsignaling polytope defined as the convex hull of 16 vertices: the 16 pairs of functions $(f_A,f_B)$ where $f_A,f_B:\{0,1\} \to \{0,1\}$. Its facets, excluding the ones from the nonsignaling polytope above are:
\begin{equation}\label{eqn:CHSH}
\begin{array}{rcl}
2 &\geq c_{0,0} + c_{0,1} - c_{1,0} + c_{1,1} &\geq -2,\\
2 &\geq c_{0,0} - c_{0,1} + c_{1,0} + c_{1,1} &\geq -2,\\
2 &\geq c_{0,0} + c_{0,1} + c_{1,0} - c_{1,1} &\geq -2,\\
2 &\geq c_{0,0} - c_{0,1} - c_{1,0} - c_{1,1} &\geq -2.
\end{array}
\end{equation}
The second of these is considered the ``usual'' form for the CHSH inequalities, although they are all equivalent under natural symmetries.

From Lemma \ref{lemma:symmetric-reduction}, the facets of the symmetric nonsignaling polytope are obtained from those of the full nonsignaling polytope by setting $a_j = b_j$, and $c_{0,1} = c_{1,0}$. This produces a 5-dimensional polyhedron in $\mathbb{R}^8$ defined by
\begin{align*}
    1 + a_{j} + a_{k} + c_{j,k} &\geq 0\\
    1 - a_{j} - a_{k} + c_{j,k} &\geq 0\\
    1 - a_{j} + a_{k} - c_{j,k} &\geq 0\\
    1 + a_{j} - a_{k} - c_{j,k} &\geq 0.
\end{align*}
Separating out the the cases $j = k$ versus $j < k$ produces
\begin{align*}
    1 + 2 a_{j} + c_{j,j} \geq 0,&\quad 1 + a_{j} + a_{k} + c_{j,k} \geq 0,\\
    1 - 2 a_{j} + c_{j,j} \geq 0,&\quad 1 - a_{j} - a_{k} + c_{j,k} \geq 0,\\
    1  - c_{j,j} \geq 0, &\quad 1 - a_{j} + a_{k} - c_{j,k} \geq 0,\\
    &\quad 1 + a_{j} - a_{k} - c_{j,k} \geq 0.
\end{align*}

Lemma \ref{lemma:symmetric-reduction} does not directly apply to the symmetric hidden variables polytope. While this polytope is contained in the symmetric nonsignaling polytope, and so satisfies $a_j = b_j$ and $c_{0,1} = c_{1,0}$, there is no reason to believe that it can be generated by intersecting the facets of the full hidden variables polytope with the space defined by these inequalities. Nonetheless, we will see this is precisely the case (after a fashion). The symmetric hidden variables polytope is defined as the convex hull of 10 vertices. Four of these are given by Alice and Bob selecting a strategy based on the same function $(f_A,f_B) = (f_j,f_j)$. The other six arise from a randomized strategy: for each pair of distinct functions $(f_j,f_k)$ with $j\not= k$, with probability $\frac{1}{2}$ they choose $(f_A,f_B) = (f_j,f_k)$ and otherwise choose $(f_A,f_B) = (f_k,f_j)$. This defines a 5-dimensional polyhedron in $\mathbb{R}^8$ whose facets (excluding the equations above) are
$$\begin{array}{rrr}
1 + 2a_{0} + c_{0,0} \geq 0, &\quad 1 + a_{0} + a_{1} + c_{0,1} \geq 0, &\quad 2 + c_{0,0} + 2c_{0,1} - c_{1,1} \geq 0,\\
1 + 2a_{1} + c_{1,1} \geq 0, &\quad 1 - a_{0} - a_{1} + c_{0,1} \geq 0, &\quad 2 + c_{0,0} - 2c_{0,1} - c_{1,1} \geq 0,\\
1 - 2a_{0} + c_{0,0} \geq 0, &\quad 1 - a_{0} + a_{1} - c_{0,1} \geq 0, &\quad 2 - c_{0,0} + 2c_{0,1} + c_{1,1} \geq 0,\\
1 - 2a_{1} + c_{1,1} \geq 0, &\quad 1 + a_{0} - a_{1} - c_{0,1} \geq 0, &\quad 2 - c_{0,0} - 2c_{0,1} + c_{1,1} \geq 0,\\
1 - c_{0,0} \geq 0, \\
1 - c_{1,1} \geq 0.
\end{array}$$
Removing the facets that are already facets of the symmetric nonsignalling polytope leaves us with the the inequalities
$$\begin{array}{rcl}
2 &\geq c_{0,0} + 2c_{0,1} - c_{1,1} &\geq -2,\\
2 &\geq c_{0,0} - 2c_{0,1} - c_{1,1} &\geq -2.
\end{array}$$
Note that if we reduce the 8 CHSH inequalities (\ref{eqn:CHSH}) by $a_j = b_j$ and $c_{0,1} = c_{1,0}$ then four of these reduce precisely to the inequalities above. However the other four collapse to the two inequalities
$$\begin{array}{rcl}
2 &\geq c_{0,0} + c_{1,1} &\geq -2,
\end{array}$$
Neither of these are facets of either the symmetric nonsignaling polytope or the symmetric hidden variables polytope. Yet they are both supporting hyperplanes of the symmetric nonsignaling polytope: one is generated from the sum of $1 - c_{0,0} \geq 0$ and $1 - c_{1,1} \geq 0$ and the other from the sum of $1 + 2a_{0} + c_{0,0} \geq 0$, $1 - 2a_{0} + c_{0,0} \geq 0$, $1 + 2a_{1} + c_{1,1} \geq 0$, and $1 - 2a_{1} + c_{1,1} \geq 0$. Consequently these inequalities are trivially satisfied for the whole symmetric nonsignaling polytope.

The situation is similar for the synchronous nonsignaling polytope. Using Lemma \ref{lemma:synchonous-reduction} we reduce the facets of the nonsignaling polytope, now by $a_j = b_j$ and $c_{j,j} = 1$. Again separating out the cases $j=k$ and $j \not= k$ we obtain inequalities
$$\begin{array}{rr}
1 + a_j \geq 0,&\quad 1 + a_{j} + a_{k} + c_{j,k} \geq 0,\\
1 - a_j \geq 0,&\quad 1 - a_{j} - a_{k} + c_{j,k} \geq 0,\\
&\quad 1 - a_{j} + a_{k} - c_{j,k} \geq 0,\\
&\quad 1 + a_{j} - a_{k} - c_{j,k} \geq 0.\\
\end{array}$$
Yet, unlike what happens in the symmetric nonsignaling polytope, not all of these are facets. The inequality $1 + a_j \geq 0$ can be realized as the sum of $1 + a_{j} + a_{k} + c_{j,k} \geq 0$ and $1 + a_{j} - a_{k} - c_{j,k} \geq 0$ and similarly for $1 + a_j \geq 0$. The remaining inequalities are independent, so the hyperplane representation of the synchronous nonsignaling polytope is
\begin{align*}
    1 + a_{j} + a_{k} + c_{j,k} &\geq 0,\\
    1 - a_{j} - a_{k} + c_{j,k} &\geq 0,\\
    1 - a_{j} + a_{k} - c_{j,k} &\geq 0,\\
    1 + a_{j} - a_{k} - c_{j,k} &\geq 0,
\end{align*}
for $j\not= k$.

While the synchronous nonsignaling polytope above is 4-dimensional, the synchronous hidden variables polytope is in fact 3-dimensional. As every synchronous classical correlation is symmetric, the polyhedron must additionally satisfy $c_{0,1} = c_{1,0}$. The polytope is generated by four vertices, the four strategies where Alice and Bob both select a function $f:\{0,1\} \to \{0,1\}$. This is just a tetrahedron and so has 4 facets as well. These facets are 
\begin{align*}
    1 + a_{0} + a_{1} + c_{0,1} &\geq 0,\\
    1 + a_{0} - a_{1} - c_{0,1} &\geq 0,\\
    1 - a_{0} - a_{1} + c_{0,1} &\geq 0,\\
    1 - a_{0} + a_{1} - c_{0,1} &\geq 0.
\end{align*}
All of these are already facets of the synchronous nonsignaling polytope, so we see there are no analogues of the CHSH inequalities for synchronous correlations. This provides an alternate proof (in the restricted case $|Y|=2$) from that of the previous appendix.
\end{section}

\begin{section}{Additional notes about \texorpdfstring{$As, Bs$ and $Ss$}{Lg}}\label{appendix:notes}

Starting with the asymmetry terms $A_{x_A,x_B}$, we use the definition of $c_{j,k}$ and $\sum_{y_A,y_B} p(y_A,y_B\:|\:x_A,x_B) = 1$ to rewrite these as:
\begin{align}
    A_{0,1} &= (p(0,1\:|\:1,0) - p(1,0\:|\:0,1)) + (p(1,0\:|\:1,0) - p(0,1\:|\:0,1)), \label{eqn:A01}\\
    A_{0,2} &= (p(0,1\:|\:2,0) - p(1,0\:|\:0,2)) + (p(1,0\:|\:2,0) - p(0,1\:|\:0,2)),\label{eqn:A02}\\
    A_{1,2} &= (p(0,1\:|\:2,1) - p(1,0\:|\:1,2)) + (p(1,0\:|\:2,1) - p(0,1\:|\:1,2)).\label{eqn:A12}
\end{align}

The bias terms $B_x$ are somewhat harder to simplify into a good form. One immediately finds
$$B_x = b_x - a_x = 2(p(0,1\:|\:x,x) - p(1,0\:|\:x,x)),$$
however it is difficult to rationalize this as a sort of bias. The key is to recognize that the bias terms always come in pairs. So if we wish to estimate, say, $B_0 - B_1$ then we first use the nonsignaling conditions to write
\begin{align*}
    B_0 &= p(0,0\:|\:0,0) + p(0,1\:|\:0,0) - p(1,0\:|\:0,0) - p(1,1\:|\:0,0)\\
    &\quad - p(0,0\:|\:0,0) + p(0,1\:|\:0,0) - p(1,0\:|\:0,0) + p(1,1\:|\:0,0)\\
    &= p(0,0\:|\:0,1) + p(0,1\:|\:0,1) - p(1,0\:|\:0,1) - p(1,1\:|\:0,1)\\
    &\quad - p(0,0\:|\:1,0) + p(0,1\:|\:1,0) - p(1,0\:|\:1,0) + p(1,1\:|\:1,0).
\end{align*}
Similar we rewrite
\begin{align*}
    B_1 &= p(0,0\:|\:1,0) + p(0,1\:|\:1,0) - p(1,0\:|\:1,0) - p(1,1\:|\:1,0)\\
    &\quad - p(0,0\:|\:0,1) + p(0,1\:|\:0,1) - p(1,0\:|\:0,1) + p(1,1\:|\:0,1).
\end{align*}
Thus
\begin{equation}
    B_0 - B_1 = 2(p(0,0\:|\:0,1) - p(1,1\:|\:0,1)) - 2(p(0,0\:|\:1,0) - p(1,1\:|\:1,0)).\label{eqn:B0minusB1}
\end{equation}
Identically,
\begin{align}
    B_0 - B_2 &= 2(p(0,0\:|\:0,2) - p(1,1\:|\:0,2)) - 2(p(0,0\:|\:2,0) - p(1,1\:|\:2,0)),\label{eqn:B0minusB2}\\
    B_1 - B_2 &= 2(p(0,0\:|\:1,2) - p(1,1\:|\:1,2)) - 2(p(0,0\:|\:2,1) - p(1,1\:|\:2,1)).\label{eqn:B1minusB2}
\end{align}
Note that these three quantities are not independent as $(B_0 - B_1) - (B_0 - B_2) + (B_1 - B_2) = 0$.

Finally the asynchronicity terms $S_x = 1 - c_{x,x}$ are again simple to express in terms of the correlation:
\begin{align}
    S_0 &= 2(p(0,1\:|\:0,0) + p(1,0\:|\:0,0)),\\
    S_1 &= 2(p(0,1\:|\:1,1) + p(1,0\:|\:1,1)),\\
    S_2 &= 2(p(0,1\:|\:2,2) + p(1,0\:|\:2,2)).
\end{align}

\begin{proposition}
    Suppose $\max\{|A_{j,k}|, |B_j|, S_j\} \leq \epsilon$. Then no synchronous Bell violation $J_3 < \frac{\epsilon}{2}$ can be explained by a (asymmetric, biased, and asynchronous) hidden variables correlation. \label{proposition:asynchronicity-bias-S-J-bounds}
\end{proposition}

During a run of the protocol, Alice and Bob can directly estimate the asymmetry and bias in their correlation using equations $\ref{eqn:A01}-\ref{eqn:B1minusB2}$. They can then use Proposition \ref{proposition:asynchronicity-bias-S-J-bounds}, which follows directly from equations (\ref{eqn:hv-inequalities}), to detect interference by Eve, and subsequently abort the protocol if their statistics don't agree with the threshold given by the proposition. As for the asynchronicity measures, these terms involve the correlation components where Alice and Bob use the same inputs. In the protocol Alice and Bob do not share their outputs in this case, and hence the asynchronicity cannot be estimated directly. However, if the observed Bell violation ($J_3$) is high, the asynchronicity is low, and Alice and Bob can apply standard error correction, information reconciliation, and privacy amplification to obtain a shared secret key with desired security parameters.

\end{section}

\begin{section}{The symmetric nonsignalling polytope with  \texorpdfstring{$|X| = 3,\ |Y| = 2$}{Lg}} \label{appendix:symnspoly}

The $21$ facets of the symmetric nonsignalling polytope reduce to simple forms: $9$ are given by
\begin{equation*}
\begin{array}{rl@{\quad}rl@{\quad}rl}
S_0 &\geq 0,& -S_0 + U_0 + 2 \geq 0,& -S_0 - U_0 + 2 \geq 0,\\
S_1 &\geq 0,& -S_1 + U_1 + 2 \geq 0,& -S_1 - U_1 + 2 \geq 0,\\
S_2 &\geq 0,& -S_2 + U_2 + 2 \geq 0,& -S_2 - U_2 + 2 \geq 0,
\end{array}
\end{equation*}
and remaining $12$ are
\begin{equation*}
\begin{array}{rl@{\quad}rl@{\quad}rl@{\quad}rl}
4J_0 + 4J_1 + U_0 - U_1 &\geq 0&, 4J_0 + 4J_1 - U_0 + U_1 &\geq 0,&
4J_0 + 4J_2 + U_0 - U_2 &\geq 0&, 4J_0 + 4J_2 - U_0 + U_2 &\geq 0,\\
4J_0 + 4J_3 - U_1 - U_1 &\geq 0&, 4J_0 + 4J_3 + U_1 + U_2 &\geq 0,&
4J_1 + 4J_2 + U_1 - U_2 &\geq 0&, 4J_1 + 4J_2 - U_1 + U_2 &\geq 0,\\
4J_1 + 4J_3 - U_0 - U_2 &\geq 0&, 4J_1 + 4J_3 + U_0 + U_2 &\geq 0,&
4J_2 + 4J_3 - U_0 - U_1 &\geq 0&, 4J_2 + 4J_3 + U_0 + U_1 &\geq 0.
\end{array}
\end{equation*}

The symmetric hidden variables polytope has $105$ facets, however only the first nine of the facets of the symmetric nonsignalling polytope are inherited (the other twelve are derivable from those listed below). There are $96$ remaining facets which define symmetric Bell inequalities. Twenty-four of these were already mentioned in the text, twelve of which form the critical relations (\ref{eqn:symmetric-Bell-1}):
\begin{equation*}
\begin{array}{rl@{\quad}rl@{\quad}rl@{\quad}rl@{\quad}rl@{\quad}rl}
4J_0 + S_0 &\geq 0,& 4J_0 + S_1 &\geq 0,& 4J_0 + S_2 &\geq 0,&
-4J_0 - S_0 + 4&\geq 0,& -4J_0 - S_1 + 4&\geq 0,& -4J_0 - S_2 + 4&\geq 0,\\
4J_1 + S_0 &\geq 0,& 4J_1 + S_1 &\geq 0,& 4J_1 + S_2 &\geq 0,&
-4J_1 - S_0 + 4&\geq 0,& -4J_1 - S_1 + 4&\geq 0,& -4J_1 - S_2 + 4&\geq 0,\\
4J_2 + S_0 &\geq 0,& 4J_2 + S_1 &\geq 0,& 4J_2 + S_2 &\geq 0,&
-4J_2 - S_0 + 4&\geq 0,& -4J_2 - S_1 + 4&\geq 0,& -4J_2 - S_2 + 4&\geq 0,\\
4J_3 + S_0 &\geq 0,& 4J_3 + S_1 &\geq 0,& 4J_3 + S_2 &\geq 0,&
-4J_3 - S_0 + 4&\geq 0,& -4J_3 - S_1 + 4&\geq 0,& -4J_3 - S_2 + 4&\geq 0.
\end{array}
\end{equation*}
Twenty-four facets provide bounds on a single $J$ term:
\begin{equation*}
\begin{array}{rl@{\quad}rl}
-8J_0 - S_0 - S_1 - 2U_0 - 2U_1 + 8 &\geq 0,&
-8J_0 - S_0 - S_1 + 2U_0 + 2U_1 + 8 &\geq 0,\\
-8J_0 - S_0 - S_2 - 2U_0 - 2U_2 + 8 &\geq 0,&
-8J_0 - S_0 - S_2 + 2U_0 + 2U_2 + 8 &\geq 0,\\
-8J_0 - S_1 - S_2 - 2U_1 + 2U_2 + 8 &\geq 0,&
-8J_0 - S_1 - S_2 + 2U_1 - 2U_2 + 8 &\geq 0,\\
-8J_1 - S_0 - S_1 - 2U_0 - 2U_1 + 8 &\geq 0,&
-8J_1 - S_0 - S_2 + 2U_0 - 2U_2 + 8 &\geq 0,\\
-8J_1 - S_0 - S_2 - 2U_0 + 2U_2 + 8 &\geq 0,&
-8J_1 - S_0 - S_1 + 2U_0 + 2U_1 + 8 &\geq 0,\\
-8J_1 - S_1 - S_2 - 2U_1 - 2U_2 + 8 &\geq 0,&
-8J_1 - S_1 - S_2 + 2U_1 + 2U_2 + 8 &\geq 0,\\
-8J_2 - S_0 - S_1 - 2U_0 + 2U_1 + 8 &\geq 0,&
-8J_2 - S_0 - S_1 + 2U_0 - 2U_1 + 8 &\geq 0,\\
-8J_2 - S_0 - S_2 - 2U_0 - 2U_2 + 8 &\geq 0,&
-8J_2 - S_0 - S_2 + 2U_0 + 2U_2 + 8 &\geq 0,\\
-8J_2 - S_1 - S_2 - 2U_1 - 2U_2 + 8 &\geq 0,&
-8J_2 - S_1 - S_2 + 2U_1 + 2U_2 + 8 &\geq 0,\\
-8J_3 - S_0 - S_1 - 2U_0 + 2U_1 + 8 &\geq 0,&
-8J_3 - S_0 - S_1 + 2U_0 - 2U_1 + 8 &\geq 0,\\
-8J_3 - S_0 - S_2 - 2U_0 + 2U_2 + 8 &\geq 0,&
-8J_3 - S_0 - S_2 + 2U_0 - 2U_2 + 8 &\geq 0,\\
-8J_3 - S_1 - S_2 - 2U_1 + 2U_2 + 8 &\geq 0,&
-8J_3 - S_1 - S_2 + 2U_1 - 2U_2 + 8 &\geq 0.
\end{array}
\end{equation*}
Twenty-four facets involve sums of two $J$s, twelve of which appear as (\ref{eqn:symmetric-Bell-2}):
\begin{equation*}
\begin{array}{rlrl}
4J_0 + 4J_1 - S_0 + S_1 &\geq 0,& 2J_0 + 2J_1 - U_0 + U_1 &\geq 0\\
4J_0 + 4J_1 + S_0 - S_1 &\geq 0,& 2J_0 + 2J_1 + U_0 - U_1 &\geq 0\\
4J_0 + 4J_2 - S_0 + S_2 &\geq 0,& 2J_0 + 2J_2 - U_0 + U_2 &\geq 0\\
4J_0 + 4J_2 + S_0 - S_2 &\geq 0,& 2J_0 + 2J_2 + U_0 - U_2 &\geq 0\\
4J_0 + 4J_3 - S_1 + S_2 &\geq 0,& 2J_0 + 2J_3 - U_1 - U_2 &\geq 0\\
4J_0 + 4J_3 + S_1 - S_2 &\geq 0,& 2J_0 + 2J_3 + U_1 + U_2 &\geq 0\\
4J_1 + 4J_2 - S_1 + S_2 &\geq 0,& 2J_1 + 2J_2 - U_1 + U_2 &\geq 0\\
4J_1 + 4J_2 + S_1 - S_2 &\geq 0,& 2J_1 + 2J_2 + U_1 - U_2 &\geq 0\\
4J_1 + 4J_3 - S_0 + S_2 &\geq 0,& 2J_1 + 2J_3 - U_0 - U_2 &\geq 0\\
4J_1 + 4J_3 + S_0 - S_2 &\geq 0,& 2J_1 + 2J_3 + U_0 + U_2 &\geq 0\\
4J_2 + 4J_3 - S_0 + S_1 &\geq 0,& 2J_2 + 2J_3 - U_0 - U_1 &\geq 0\\
4J_2 + 4J_3 + S_0 - S_1 &\geq 0,& 2J_2 + 2J_3 + U_0 + U_1 &\geq 0.
\end{array}
\end{equation*}
The final twenty-four facets involve sums of three $J$s.
\begin{align*}
8J_0 + 2J_1 + 2J_2 - S_0 + S_1 + S_2 + 2U_0 - U_1 - U_2 &\geq 0\\
8J_0 + 2J_1 + 2J_2 - S_0 + S_1 + S_2 - 2U_0 + U_1 + U_2 &\geq 0\\
8J_0 + 2J_1 + 2J_3 + S_0 - S_1 + S_2 + U_0 - 2U_1 - U_2 &\geq 0\\
8J_0 + 2J_1 + 2J_3 + S_0 - S_1 + S_2 - U_0 + 2U_1 + U_2 &\geq 0\\
8J_0 + 2J_2 + 2J_3 + S_0 + S_1 - S_2 + U_0 - U_1 - 2U_2 &\geq 0\\
8J_0 + 2J_2 + 2J_3 + S_0 + S_1 - S_2 - U_0 + U_1 + 2U_2 &\geq 0\\
2J_0 + 8J_1 + 2J_2 + S_0 - S_1 + S_2 - U_0 + 2U_1 - U_2 &\geq 0\\
2J_0 + 8J_1 + 2J_2 + S_0 - S_1 + S_2 + U_0 - 2U_1 + U_2 &\geq 0\\
2J_0 + 8J_1 + 2J_3 - S_0 + S_1 + S_2 + 2U_0 - U_1 + U_2 &\geq 0\\
2J_0 + 8J_1 + 2J_3 - S_0 + S_1 + S_2 - 2U_0 + U_1 - U_2 &\geq 0\\
8J_1 + 2J_2 + 2J_3 + S_0 + S_1 - S_2 - U_0 + U_1 - 2U_2 &\geq 0\\
8J_1 + 2J_2 + 2J_3 + S_0 + S_1 - S_2 + U_0 - U_1 + 2U_2 &\geq 0\\
2J_0 + 2J_1 + 8J_2 + S_0 + S_1 - S_2 + U_0 + U_1 - 2U_2 &\geq 0\\
2J_0 + 2J_1 + 8J_2 + S_0 + S_1 - S_2 - U_0 - U_1 + 2U_2 &\geq 0\\
2J_0 + 8J_2 + 2J_3 - S_0 + S_1 + S_2 + 2U_0 + U_1 - U_2 &\geq 0\\
2J_0 + 8J_2 + 2J_3 - S_0 + S_1 + S_2 - 2U_0 - U_1 + U_2 &\geq 0\\
2J_1 + 8J_2 + 2J_3 + S_0 - S_1 + S_2 - U_0 - 2U_1 + U_2 &\geq 0\\
2J_1 + 8J_2 + 2J_3 + S_0 - S_1 + S_2 + U_0 + 2U_1 - U_2 &\geq 0\\
2J_0 + 2J_1 + 8J_3 + S_0 + S_1 - S_2 - U_0 - U_1 - 2U_2 &\geq 0\\
2J_0 + 2J_1 + 8J_3 + S_0 + S_1 - S_2 + U_0 + U_1 + 2U_2 &\geq 0\\
2J_0 + 2J_2 + 8J_3 + S_0 - S_1 + S_2 + U_0 + 2U_1 + U_2 &\geq 0\\
2J_0 + 2J_2 + 8J_3 + S_0 - S_1 + S_2 - U_0 - 2U_1 - U_2 &\geq 0\\
2J_1 + 2J_2 + 8J_3 - S_0 + S_1 + S_2 + 2U_0 + U_1 + U_2 &\geq 0\\
2J_1 + 2J_2 + 8J_3 - S_0 + S_1 + S_2 - 2U_0 - U_1 - U_2 &\geq 0.
\end{align*}

\end{section}

\end{document}